\documentclass[lettersize,journal]{IEEEtran}

\usepackage[T1]{fontenc}
\usepackage[utf8]{inputenc}
\usepackage{amsmath, amssymb, amsfonts} % For advanced math
\usepackage{graphicx}                   % For including figures
\usepackage{cite}                       % For handling citations like [1]-[3]
\usepackage[linesnumbered,ruled,vlined]{algorithm2e}
\usepackage[noend]{algpseudocode}       % For writing pseudocode
\usepackage{lipsum}                     % For generating dummy text
\usepackage{hyperref}                   % For clickable links in the PDF
\usepackage{amsthm}
\usepackage{graphicx}      % To include images
\usepackage{subcaption}    % For subfigures (modern and recommended)
\usepackage{cite}
\usepackage{caption}
\usepackage{booktabs}
\usepackage{textcomp}
\usepackage{url}
\usepackage{verbatim}
\usepackage{array}

\usepackage{stfloats} % 【关键宏包】用于支持双栏底部float
\usepackage{lipsum} % 用于生成测试文本
\theoremstyle{plain}\newtheorem{theorem}{Theorem}

\newtheorem{assumption}{Assumption}
\newtheorem{proposition}{Proposition}

\theoremstyle{definition}

\hypersetup{
  colorlinks=true,
  linkcolor=blue,
  filecolor=magenta,
  urlcolor=cyan,
  pdftitle={DAFT-s-AFDM AF Paper},
}

\begin{document}

% =========================================================================
% TITLE AND AUTHOR INFORMATION
% =========================================================================
\title{\huge{DAFT-s-AFDM Enabled ISAC Systems: Ambiguity Function Analysis and Waveform Design}}

% Replace with your author information
\author{Shiqi Cui, Tianqi Mao, \textit{Member, IEEE}, Fan Zhang, Zeping Sui, \textit{Member, IEEE}, Christos Masouros, \textit{Fellow, IEEE} \\ and Zhaocheng Wang, \textit{Fellow, IEEE}%

\thanks{This work was supported by National Natural Science Foundation of China under Grant 62401054. \emph{(Corresponding authors: Tianqi Mao.)} The authors would like to thank Prof. Zilong Liu for his insightful discussions and valuable suggestions that have significantly improved the quality of this paper.}
\thanks{S. Cui and T. Mao are with State Key Laboratory of CNS/ATM, Beijing Institute of Technology, Beijing 100081, China. T. Mao is also with Beijing Institute of Technology (Zhuhai), Zhuhai 519088, China (e-mails: cuisqqq@bit.edu.cn, maotq@bit.edu.cn).

    F. Zhang and Z. Wang are with Department of Electronic Engineering, Tsinghua University, Beijing 100084, China. Z. Wang is also with Tsinghua Shenzhen International Graduate School, Shenzhen 518055, China (e-mails: zf22@mails.tsinghua.edu.cn, zcwang@tsinghua.edu.cn).

    Z. Sui is with the School of Computer Science and Electronic Engineering, University of Essex, CO4 3SQ Colchester, U.K. (e-mail: zepingsui@outlook.com).

    Christos Masouros is with the Department of Electronic and Electrical Engineering, University College London, Torrington Place, London, WC1E 7JE, UK (e-mail: c.masouros@ucl.ac.uk).

}}

\markboth{Journal of \LaTeX\ Class Files,~Vol.~14, No.~8, August~2021}%
{Shell \MakeLowercase{\textit{et al.}}: A Sample Article Using IEEEtran.cls for IEEE Journals}


% Generate the title
\maketitle

% =========================================================================
% ABSTRACT AND KEYWORDS
% =========================================================================
\begin{abstract}
 Discrete affine Fourier transform spread affine frequency division multiplexing (DAFT-s-AFDM) is a promising waveform for integrated sensing and communication (ISAC) due to its low peak-to-average power ratio, robustness to Doppler shifts, and reduced multiuser interference in the uplink transmission. This paper presents a comprehensive ambiguity function (AF) analysis of DAFT-s-AFDM and derives the closed-form expression for the AF magnitude expectation. Several key insights into the impact of DAFT-s-AFDM parameters on ISAC performance are revealed, thus providing concrete guidance for the subsequent waveform design. Building on these insights, a novel probabilistic constellation shaping (PCS) framework is proposed for ISAC waveform enhancement, where the communication throughput and the sensing AF characteristics are jointly optimized by addressing a multi-objective problem. An efficient algorithm based on a closed-form bit error rate expression is developed to obtain the Pareto-optimal solutions. Extensive simulations validate the theoretical results and that the proposed PCS-enhanced DAFT-s-AFDM can significantly outperform the classical counterparts, achieving a superior and highly controllable tradeoff between the dual-functional performances.
\end{abstract}


% =========================================================================
% I. INTRODUCTION
% =========================================================================
\section{Introduction}
Integrated sensing and communication (ISAC) is emerging as a foundational technology for future wireless networks, as it intrinsically integrates communication and sensing functionalities within a unified framework~\cite{ISAC_Survey}. Such a paradigm offers higher spectral efficiency, lower hardware cost, and the potential for mutual enhancement between communication and sensing performances~\cite{ISAC_Chirp_Sui}. In this context, it is essential to develop highly efficient and flexible waveforms so as to strike an appropriate tradeoff between the two functionalities.

Most prior ISAC waveform research has focused on adapting legacy waveforms from conventional communication or sensing systems~\cite{ISAC_Survey}. Legacy orthogonal frequency division multiplexing (OFDM) is one of the most representative ISAC waveforms~\cite{cross-domain,  CP-OFDM, Reshaping}. However, the applications of OFDM may be challenging due to its sensitivity to high inter-carrier interference (ICI) incurred by Doppler, carrier frequency offset, and phase noise~\cite{THz_HWI_Survey}. As a remedy, affine frequency division multiplexing (AFDM), has attracted tremendous research attention recently due to its excellent backwards-compatibility to OFDM and its strong resilience in doubly selective channels~\cite{AFDM, HWI_AFDM_Sui, SM_AFDM_Sui}. The basic idea of AFDM is to modulate data symbols upon multiple orthogonal chirp subcarriers which can be efficiently implemented through inverse discrete affine frequency transform (IDAFT). By properly tuning the chirp rate according to the maximum Doppler, AFDM is able to achieve full channel diversity~\cite{non_afdm_Sui}.

AFDM-enabled ISAC has been studied from multiple perspectives~\cite{afdm_isac, AFDM_AF, AFDM_AF_FANLIU, AFDM_ISAC_WCL, AFDM_Fan}. Specifically, the closed-form Cramér-Rao bounds and ambiguity function (AF) of pilot-assisted AFDM with superimposed pilot design were derived in~\cite{afdm_isac}. In~\cite{AFDM_AF}, the authors characterized both the local spike-like and global periodic AF structures of AFDM, and further introduced an unambiguity parallelogram for waveform parameter design. More recently, pulse-shaped random signaling was investigated in~\cite{AFDM_AF_FANLIU}, where controllable sidelobe suppression was demonstrated through proper waveform parameter design. From a system-level perspective, a single DAFT-domain pilot was introduced in~\cite{AFDM_ISAC_WCL} to enable range-velocity estimation and low-complexity analog self-interference cancellation, while the sensing spectral-efficiency/outage tradeoff was characterized in~\cite{AFDM_Fan}, together with an estimator capable of extending the unambiguous Doppler range.

Despite these advances, AFDM still faces several practical limitations, most notably its high peak-to-average power ratio (PAPR). In addition, practical ISAC deployments are inherently multi-user in nature, and imperfect orthogonality due to Doppler and asynchronous transmission may lead to non-negligible multi-user interference (MUI), thereby degrading communication reliability and limiting system scalability. To address these issues, DAFT-spread AFDM (DAFT-s-AFDM) was proposed in~\cite{DAFT-s-AFDM}. Motivated by its favorable communication performance, this work further investigates the potential of DAFT-s-AFDM for ISAC applications.

In parallel with the development of AFDM, probabilistic constellation shaping (PCS) has been extensively studied as a classical waveform design approach for communication systems~\cite{PCS_TWC}. PCS has already been commercialized in optical/fiber communication systems and is being progressively incorporated into fifth-generation New Radio (5G-NR), owing to its ability to provide asymptotic shaping gains, improve block error rate (BLER)-constrained throughput, and enable flexible rate adaptation with only marginal transceiver overhead~\cite{PCS_JLT, PCS_TWC}. Its original design objective is to enhance the achievable rate by shaping symbol probabilities toward a Gaussian-like distribution, such as Maxwell-Boltzmann (MB) shaping, thereby approaching channel capacity.

Very recently, PCS has been introduced into ISAC waveform design to enable flexible communication-sensing trade-offs. Representative studies have shown that PCS can render both pilot and data symbols sensing-aware while preserving communication performance, and can effectively reshape ambiguity-function sidelobes and adjust the rate-sensing operating points~\cite{Reshaping, PCS_ISAC_TWC, Constellation_TVT, PCS_Filter}. For instance, the authors of~\cite{Reshaping} investigated the impact of the normalized fourth-order moment of the constellation on the AF of OFDM waveforms, and proposed a modified Blahut-Arimoto (MBA) algorithm combined with a heuristic optimization strategy to jointly optimize communication and sensing performance over additive white Gaussian noise (AWGN) channels. This research direction was further advanced in~\cite{Constellation_TVT}, where an end-to-end autoencoder was developed for uplink OFDM-ISAC systems to achieve controllable ISAC trade-offs. Building upon these OFDM-oriented studies, the authors of~\cite{PCS_Filter} proposed a unified PCS framework for matched and mismatched filtering, in which the achievable rate is maximized under sensing-metric, power, and distribution constraints, thereby demonstrating the flexibility of PCS for ISAC trade-off design.

Nevertheless, a direct extension of conventional communication-oriented PCS schemes to ISAC systems is generally suboptimal. Most existing ISAC-oriented PCS designs still rely on a strict unit-power assumption, leaving the power-domain design space insufficiently explored. Moreover, MB-type Gaussian shaping tends to increase the PAPR of single-carrier waveforms and to introduce excessive randomness, both of which may impair sensing accuracy~\cite{ISAC_TIT, Sensing_with_commun}. This phenomenon reveals an inherent tradeoff in ISAC waveform design between communication-oriented stochastic shaping and sensing-oriented waveform determinism. Furthermore, DAFT-s-AFDM exhibits channel coupling characteristics and waveform structures that differ fundamentally from those of OFDM, thereby motivating the development of a dedicated PCS framework for DAFT-s-AFDM-based ISAC systems.

Against the above background, the primary contributions of this study are summarized as follows.

\begin{itemize}
\item A rigorous theoretical foundation for sensing with DAFT-s-AFDM waveforms is established, where the analytical AF is derived for both localized and interleaved subcarrier mappings. Specifically, we have the following observations: 1) under full subcarrier occupancy, the main sidelobes form an approximately flat pedestal; 2) the main sidelobes are formed by a weighted superposition of overlapping sinc-like components in the localized FDMA case; and 3) under interleaved FDMA mapping, one can observe a comb-like sidelobe, and periodically distributed pseudo-peaks. It is further shown that velocity-estimation accuracy is closely linked to the normalized fourth-order constellation moment.
\item PCS design is formulated as a novel multi-objective optimization problem that explicitly characterizes the trade-off between communication rate and velocity sensing performance. To achieve Pareto-optimal performance, a low-complexity algorithm is developed based on a closed-form expressions of approximated bit error rate (BER), enabling analytically evaluation of the objective, identification of Pareto front with reduced computational burden and faster runtime.
\item Simulation results validate the theoretical findings and demonstrate under doubly selective channels, PCS-based DAFT-s-AFDM achieves a significantly improved and controllable trade-off between throughput and sensing performance over baselines, thereby quantitatively confirming the benefits of shaping-aware waveform design.

\end{itemize}

The remainder of this paper is organized as follows. Section~\ref{sec:system_model} introduces the system model. Section~\ref{sec:sensing_performance} details the derivation and analysis of the AF. Section~\ref{sec:comm_performance} formulates the communication performance metric. Section~\ref{sec:optimization} presents the proposed PCS optimization algorithms. Numerical results are provided in Section~\ref{sec:results}, and Section~\ref{sec:conclusion} concludes the paper.

\begin{figure*}[t!]
\centering
\includegraphics[width=0.85\linewidth]{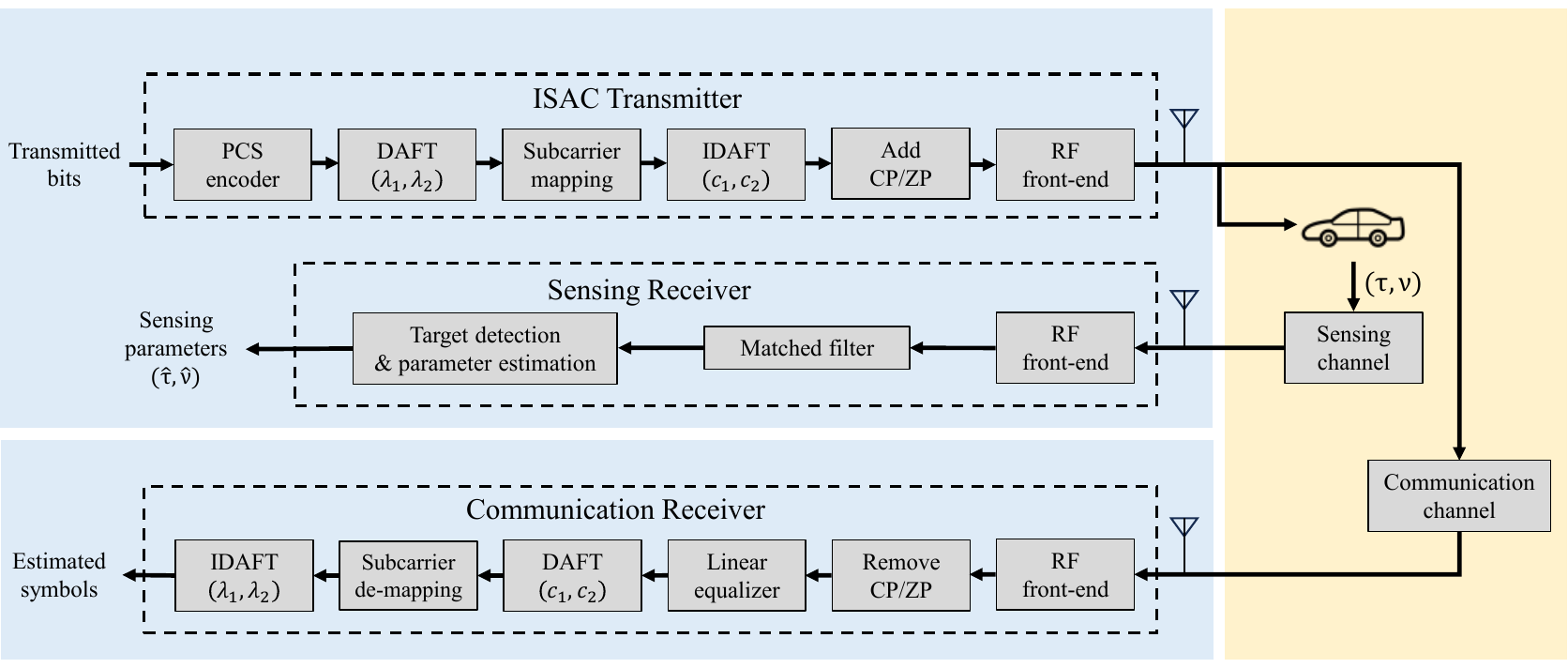}
\caption{Transceiver diagram of the DAFT-s-AFDM-enabled monostatic ISAC system.}
\label{fig:system_model}
% \vspace{-3mm}
\end{figure*}

% =========================================================================
% II. SYSTEM MODEL
% =========================================================================
\textit{Notations:} Bold lowercase and uppercase letters denote vectors and matrices. $(\cdot)^T$, $(\cdot)^H$, and $(\cdot)^*$ denote transpose, Hermitian transpose, and complex conjugate. $\mathbb{C}$ is the complex field; $\mathbf{I}_N$ and $\mathbf{0}$ denote the identity and zero matrices. $\mathrm{diag}(\cdot)$ and $\otimes$ denote diagonalization and Kronecker product. $\mathbb{E}[\cdot]$ is expectation, and $\mathrm{Re}\{\cdot\}$/$\mathrm{Im}\{\cdot\}$ are real/imaginary parts. $\langle a \rangle_N$ is modulo-$N$, and $\lfloor \cdot \rfloor$, $\lceil \cdot \rceil$ are floor/ceiling operators.

\section{System Model}
Fig.~\ref{fig:system_model} illustrates a DAFT-s-AFDM-enabled monostatic ISAC system, in which a unified transceiver performs data transmission and target sensing using a shared waveform, with the sensing transmitter and receiver being co-located. At the transmitter, the input bits are first passed through a PCS encoder to generate constellation symbols following a prescribed probability mass function (PMF), and these shaped symbols are then modulated into a DAFT-s-AFDM waveform. At the receiver, the sensing receiver processes the received echo through matched filtering followed by target detection and parameter estimation to extract the target range and Doppler information, whereas the communication receiver performs linear equalization and symbol demodulation to recover the estimated symbols over the doubly selective channel.
\label{sec:system_model}

\subsection{DAFT-s-AFDM}
\label{subsec:DAFT-s-AFDM}
Without loss of generality, we consider a power-normalized square quadrature amplitude modulation (QAM) alphabet $\mathcal{X}$. Let $X\in\mathcal{X}$ denote the shaped symbol with PMF $P_X(x)=p_x$, constrained by
\begin{equation}
p_x\ge 0,\quad \sum_{x\in\mathcal{X}} p_x=1,\quad \frac{1}{|\mathcal{X}|}\sum_{x\in\mathcal{X}} |x|^2=1.
\end{equation}
% We realize $P_X$ via constant-composition distribution matching (CCDM)~\cite{CCDM}. For each length-$M$ block, CCDM enforces a fixed symbol composition whose empirical distribution is close to the target PMF and maps bits to sequences with that composition. The distribution-matching rate is determined by the number of admissible sequences under the selected composition and approaches the entropy-limited shaping rate as $M$ increases. Thus, for sufficiently large $M$, the CCDM output distribution approaches the target PMF~\cite{CCDM}. This symbol-domain PMF is then embedded into the DAFT-s-AFDM waveform through the following transceiver processing.
The target PMF $P_X$ is implemented by constant-composition distribution matching (CCDM)~\cite{CCDM}, which maps input bits to symbol sequences with a prescribed composition so that the empirical distribution approximates $P_X$. The resulting symbol-domain PMF is then embedded into the DAFT-s-AFDM waveform through the following transceiver processing.

Consequently, the DAFT-s-AFDM transmitter employs a unified matrix representation to map a symbol block $\mathbf{x} \in \mathbb{C}^{M\times1}$ containing $M$ information symbols into an $N$-point time-domain signal $\mathbf{s} \in \mathbb{C}^{N\times1}$. The process begins with DAFT spreading, where an $M$-point DAFT transforms the input vector $\mathbf{x}$ into the DAF domain, yielding
\begin{equation}
\mathbf{x}_{\text{DAF}} = \mathbf{A}_{\lambda_1, M} \mathbf{F}_M \mathbf{A}_{\lambda_2, M} \mathbf{x},
\end{equation}
where $\mathbf{F}_M$ denotes the $M$-point normalized DFT matrix, and the $M \times M$ diagonal chirp matrix parameterized by $\alpha$ is defined as
\begin{equation}
\mathbf{A}_{\alpha, M} = \text{diag}\left(1, e^{-j2\pi \alpha}, \dots, e^{-j2\pi \alpha (M-1)^2}\right).
\end{equation}

The spread symbols are subsequently adapted to the full transmission block size via the mapping matrix $\mathbf{\Gamma}$, an $N \times M$ subcarrier mapping matrix that positions $M$ spread symbols on $N$ available chirp subcarriers at a uniform spacing $S$. Its structure is given by $\mathbf{\Gamma} = [\mathbf{I}_M \otimes \mathbf{u}, \mathbf{0}_{M \times (N - M \cdot S)}]^T$, where $\otimes$ signifies the Kronecker product, $\mathbf{u} = [1, 0, \dots, 0]$ is a length-$S$ base vector, and $\mathbf{\Gamma}$ pads the remaining unused subcarriers to achieve dimension $N$. This mapping places the $M$ spread symbols at $M$ selected positions among the $N$ available subcarriers while zero-padding the unoccupied indices. The signal is then transformed into the time domain by the AFDM modulation module, which implements an $N$-point inverse DAFT (IDAFT) via $\mathbf{A}_{c_1, N}^H \mathbf{F}_N^H \mathbf{A}_{c_2, N}^H$. The resulting time-domain signal $\mathbf{s}$ is expressed as
\begin{align}
\mathbf{s} &=  \mathbf{A}_{c_1, N}^H \mathbf{F}_N^H \mathbf{A}_{c_2, N}^H \mathbf{\Gamma} \mathbf{x}_{\text{DAF}} \nonumber \\
&\triangleq \mathbf{U}_{N, M, c_1, c_2,\lambda_1, \lambda_2}\mathbf{x}.
\end{align}
Finally, a chirp-periodic prefix (CPP) or zero padding (ZP) is prepended to the signal to deal with the multipath propagation \cite{DAFT-s-AFDM, TZP_AFDM, TZP_AFDM_Sui}. It is noteworthy that when $2Nc_1 \in \mathbb{Z}$, CPP becomes equivalent to the conventional cyclic prefix (CP). All subsequent analyses in this paper are based on this assumption.

\subsection{Channel Model}
Both the communication and sensing channels are modeled as doubly selective channels. Let $T_s$ denote the sampling interval and $N$ the number of samples in each frame. The channel impulse response in the delay-Doppler domain is given by~\cite{AFDM, AFDM_AF}
\begin{equation}
\label{eq:channel_model_discrete}
h[l,k] = \sum_{i=1}^{P} h_i \delta[l-l_i]\delta[k-k_i],
\end{equation}
where $l$ and $k$ denote the delay and Doppler indices, respectively, with $\tau=lT_s$ and $\nu=k/(NT_s)$, where $\tau$ and $\nu$ represent the delay and Doppler variables. Here, $P$ denotes the number of propagation paths (or targets in the sensing scenario), $h_i$ is the complex gain of the $i$th path, and $l_i$ and $k_i$ are the corresponding normalized integer delay and Doppler taps, respectively. Moreover, $\delta[\cdot]$ denotes the discrete Kronecker delta function.

Subsequently, the received discrete-time signal, $r[n]$, at either the communication or sensing receiver can be formulated by summing the transmitted signal $s[n]$ over the delay-Doppler taps, which can be expressed as
\begin{equation}
r[n] = \sum_{i=1}^{P} h_i s[(n - l_i)_N] e^{j \frac{2\pi k_i n}{N}} + w[n],
\end{equation}
where the discrete-time AWGN satisfying $w[n] \sim \mathcal{CN}(0, \sigma^2)$, with $\sigma^2$ denoting the noise variance.

% =========================================================================
% III. FORMULATION OF SENSING PERFORMANCE VIA AMBIGUITY FUNCTION
% =========================================================================
\section{Ambiguity Function Analysis for the DAFT-s-AFDM Waveform}
\label{sec:sensing_performance}
This section presents the theoretical framework for characterizing the sensing performance of the DAFT-s-AFDM. We first derive the analytical expression for its AF. The expected value of this function is then analyzed to characterize the fundamental relationship among waveform parameters, constellation statistics, and sensing performance.

\subsection{AF Derivations}
The AF describes the response of a matched filter to the signal with varying time delays $\tau$ and Doppler shifts $\nu$. To accurately characterize this response, it is essential to distinguish between the signal processing models dictated by the presence CPs or ZPs.

In OFDM and AFDM, the inclusion of a CP is standard practice to eliminate inter-symbol interference (ISI). This configuration effectively treats the signal as periodic, implying that the matched filtering at the sensing receiver corresponds to a periodic convolution. Conversely, long-range detection tasks of sensing where targets may lie beyond the CP coverage, or in scenarios requiring a sufficiently small duty cycle, the ISAC signal operates as a zero-padded waveform. In such cases without a CP, the matched filter processing is governed by linear convolution. Accordingly, to present a thorough analysis covering both ISAC signaling cases, we define two distinct forms of the discrete AF~\cite{CP-OFDM}.

For the case with CP, the sensing performance is characterized by the discrete periodic AF (DPAF), denoted as $\chi_P$, which can be expressed by \cite{CP-OFDM}
\begin{equation}
\label{eq:af_periodic_def}
\chi_{P}(\tau, \nu) = \sum_{n=0}^{N-1} s[n] s^*[(n-\tau)_N] e^{-j\frac{2\pi}{N}\nu n},
\end{equation}
where $(\cdot)_N$ denotes the modulo $N$ operation.

For the case with ZP, the sensing performance is characterized by the discrete aperiodic AF (DAAF), denoted as $\chi_A$. Here, the summation corresponds to linear convolution and is restricted to the valid overlap interval
\begin{equation}
\label{eq:af_aperiodic_def}
\chi_{A}(\tau, \nu) =
\begin{cases}
\displaystyle \sum_{n=0}^{N-1-\tau} s[n] s^*[n-\tau] e^{-j\frac{2\pi}{N}\nu n}, & 0 \le \tau < N, \\
\displaystyle \sum_{n=-\tau}^{N-1} s[n] s^*[n-\tau] e^{-j\frac{2\pi}{N}\nu n}, & -N < \tau < 0.
\end{cases}
\end{equation}

Based on the matrix-form signal model shown in Sec. \ref{sec:system_model}, the DAFT-s-AFDM signal without prefix can be further expressed in element-wise form as
\begin{equation}
\label{eq:sn}
s[n] = \sum_{m = 0}^{M-1} x_m g_m[n],
\end{equation}
where $g_m[n]$ is given by
\begin{equation}
\begin{split}
\label{eq:gm_orign}
g_m[n] &= \frac{e^{j2\pi (c_1 n^2 - \lambda_1 m^2)}}{\sqrt{NM}} \sum_{l=0}^{M-1} e^{ j2\pi \left[\Delta \lambda l^2 + \left(S \frac{n}{N} - \frac{m}{M}\right)l \right]},
\end{split}
\end{equation}
where $\Delta \lambda = c_2 S^2 - \lambda_2$, $T$ is the symbol duration and $S$ denotes the subcarrier spacing factor, with $S=1$ for localized-frequency division multiple access (L-FDMA) and $S=N/M$ for interleaved-FDMA (I-FDMA). Let $\chi(\tau, \nu)$ denote the AF representing either $\chi_P$ or $\chi_A$. By substituting the signal model \eqref{eq:sn} into the general AF definition which representing either \eqref{eq:af_periodic_def} or \eqref{eq:af_aperiodic_def}, we obtain a unified expression
\begin{align}
\chi(\tau, \nu) &= \sum_{n \in \mathcal{I}} \left( \sum_m x_m g_m[n] \right) \left( \sum_p x_p^* g_p^*[n'] \right) e^{-j\frac{2\pi}{N}\nu n} \nonumber\\
&= \sum_m \sum_p x_m x_p^* \underbrace{\left( \sum_{n \in \mathcal{I}} g_m[n] g_p^*[n'] e^{-j\frac{2\pi}{N}\nu n} \right)}_{\mathcal{A}_{g_m,g_p}(\tau, \nu)},
\label{eq:chi_general_form}
\end{align}
where the specific definitions of the summation domain $\mathcal{I}$ and the time-shifted index $n'$ are dictated by the chosen AF formulation. For the periodic AF ($\chi_P$), due to the cyclic property, the summation is performed over the full length $N$ regardless of delay $\tau$. The summation domain is
\begin{equation}
\mathcal{I}_{\tau, \text{Per}} = [0, N-1], \quad \forall \tau \in \mathbb{Z}.
\end{equation}
By contrast, for $\tau \neq 0$, the expected sidelobe level remains constant with respect to $\tau$. For the aperiodic AF ($\chi_A$), the summation is restricted to the non-zero overlap duration $L_{\tau} = N - |\tau|$, and the domain is defined as
\begin{equation}
\mathcal{I}_{\tau, \text{Aper}} =
\begin{cases}
\left[\tau, N-1 \right], & \tau \ge 0 \\
[0, N-1-|\tau|], &  \tau < 0
\end{cases}
\end{equation}
This leads to a "roof-top" decay in the expected sidelobe level as $|\tau|$ increases.

Subsequently, taking the expectation of the squared modulus of $\chi(\tau, \nu)$ yields
\begin{equation}
\begin{split}
&E\left\{|\mathcal{\chi}(\tau, \nu)|^2\right\} = E\left[ \left| \sum_{m,p} x_m x_p^* \mathcal{A}_{g_m, g_p}(\tau, \nu) \right|^2 \right]  \\
&~\quad= \sum_{m,p,m',p'} \mathcal{A}_{g_m, g_p} \mathcal{A}_{g_{m'}, g_{p'}}^{*} E[x_m x_p^* x_{m'}^* x_{p'}].
\label{eq:proof_quad_sum_appendix}
\end{split}
\end{equation}

To derive the statistical properties of the AF, the following assumptions concerning the constellation symbols $x$ are introduced.
\begin{assumption}
\label{ass:symbol_stats}
The transmitted symbol sequence $\{x_m\}$ is modeled as a set of zero-mean, mutually independent random variables. Specifically, the first-order moment is assumed to satisfy $\mathbb{E}[x_m] = 0$, and for distinct indices $m \neq p$, the orthogonality condition $\mathbb{E}[x_m x_p^*] = 0$ is maintained.
\end{assumption}

Drawing upon Assumption \ref{ass:symbol_stats}, the non-zero terms in \eqref{eq:proof_quad_sum_appendix} arise only when indices are matched ($m=p=m'=p'$ for 4th order and pairs $\{m=p, m'=p'\}$ and $\{m=p', p=m'\}$ for 2nd order), $\mathbb{E}\left[|\chi(\tau,\nu)|^2\right]$ can be further simplified as
\begin{equation}
\begin{split}
\mathbb{E}\left\{|\chi(\tau, \nu)|^2\right\} &= \ \left| \sum_{m=0}^{M-1} \mathcal{A}_{g_m, g_m}(\tau, \nu) \right|^2  \\
& +  \sum_{m=0}^{M-1} \sum_{p=0}^{M-1} |\mathcal{A}_{g_m, g_p}(\tau, \nu)|^2 \\
& + (\mu_4 - 2) \sum_{m=0}^{M-1} |\mathcal{A}_{g_m, g_m}(\tau, \nu)|^2.
\label{eq:expected_af_general}
\end{split}
\end{equation}
where $\mu_4 = \frac{\mathbb{E}_4}{\sigma_x^4}$ denotes the normalized fourth-order moment of the constellation symbols with the average symbol power $\sigma_x^2$ and $\mathbb{E}_4 = \mathbb{E} \{|x|^4\}$ is the fourth-order moment.

% ----------- 我是定理1 --------------
\begin{figure*}[b] % [b] 表示强制置于底部
\hrulefill

\begin{equation}
\begin{split}
\label{eq:DAFT_AF}
&\mathbb{E} \left\{ |\chi(\tau, \nu)|^2 \right\} =  \underbrace{ \frac{1}{N^2} \mathcal{D}_M\left(\frac{S\tau}{N}\right) \cdot \mathcal{D}_{L_\tau}\left(\frac{\phi}{N}\right)}_{\mathcal{T}_1}
+ \underbrace{\frac{1}{N^2} \sum_{k=-(L_\tau-1)}^{L_\tau-1} (L_\tau - |k|) \cdot \mathcal{D}_M\left(\frac{Sk}{N}\right) \cdot \cos\left( \frac{2\pi k}{N} \phi \right)}_{\mathcal{T}_2} \\
&+ \left(\mu_4 - 2 \right) \cdot \underbrace{\frac{1}{N^2 M} \left( \sum_{k=-(M-1)}^{M-1}\mathcal{D}_{L_\tau}(\xi_k) \cdot \mathcal{D}_{M-|k|}(\eta_k) + 2\sum_{k=1}^{M-1} \left[ \mathcal{S}_{L_\tau}(\xi_k) \mathcal{S}_{L_\tau}(\xi_{k-M}) \right] \cdot \left[ \mathcal{S}_{M-k}(\eta_k) \mathcal{S}_{k}(\eta_{k-M}) \right] \cdot \cos(\Psi) \right)}_{\mathcal{T}_3}.
\end{split}
\end{equation}
\end{figure*}
% ----------- 我是定理1 --------------

Subsequently, by defining $\mathcal{T}_\text{1} \triangleq \left| \sum_{m=0}^{M-1} \mathcal{A}_{g_m, g_m}(\tau, \nu) \right|^2$, $\mathcal{T}_\text{2} \triangleq \sum_{m=0}^{M-1} \sum_{p=0}^{M-1} |\mathcal{A}_{g_m, g_p}(\tau, \nu)|^2$, and $\mathcal{T}_\text{3} \triangleq  \sum_{m=0}^{M-1} |\mathcal{A}_{g_m, g_m}(\tau, \nu)|^2 $ respectively, we obtain the following theorem.
% Consequently, we present the following theorem.

\begin{theorem}[Closed-Form Expressions for AF]
\label{theorem:AF_DAFT}
The expectation of the squared modulus of the DAFT-s-AFDM AF is given by \eqref{eq:DAFT_AF}, where $\phi \triangleq 2Nc_1\tau - \nu $, $\xi_k \triangleq \frac{1}{N}(2Nc_1\tau + Sk - \nu) $, $ \eta_k \triangleq \frac{S\tau}{N} + 2k\Delta\lambda $, $\mathcal{S}_L(x) \triangleq \frac{\sin(\pi L x)}{\sin(\pi x)}$, $ \mathcal{D}_L(x) \triangleq \left| \mathcal{S}_L(x) \right|^2$ and 
\begin{align}
\Psi = 
\begin{cases}
\frac{\pi SM(N-1-\tau)}{N}
+ 2\pi M(M-1)\Delta\lambda,
& \text{for the DPAF},\\
\frac{\pi SM(N-1)}{N}
+ 2\pi M(M-1)\Delta\lambda,
& \text{for the DAAF}.
\end{cases}
\notag
\end{align}
The effective summation length $L_\tau$ is defined as
\begin{equation}
L_\tau =
\begin{cases}
N, & \text{for the DPAF}, \\
N - |\tau|, & \text{for the DAAF}.
\end{cases}
\label{eq:L_tau_def}
\end{equation}

\end{theorem}
\begin{proof}
See \textbf{App.}~A in the supplemental document \cite{proof} (also attached with the main manuscript).
\end{proof}

\subsection{Properties Analysis}
\label{sec:properties}
\textbf{Theorem} \ref{theorem:AF_DAFT} facilitates a further analysis of the properties of the DAFT-s-AFDM AF. The impact of the normalized kurtosis $\mu_4$ of the constellation on the AF is examined first. Specifically, the following proposition holds.

\begin{proposition}
\label{prop:1}
Under the condition $\Delta \lambda \in \{n/2 \mid n \in \mathbb{Z}\}$, the normalized kurtosis $\mu_4$ of the constellation primarily influences the AF at points near $\tau = \frac{kN}{S}$, where $k \in \mathbb{Z}$.
\end{proposition}
\begin{proof}
See \textbf{App.}~B in the supplemental document \cite{proof}.
\end{proof}

\begin{proposition}
\label{prop:2}
The expected value of the DAFT-s-AFDM AF at the origin is
\begin{equation}
\label{eq:AF_00}
\mathbb{E} \left\{ |\chi(0, 0)|^2 \right\} = M^2 + (\mu_4 - 1)M.
\end{equation}

\end{proposition}
\begin{proof}
See \textbf{App.}~C in the supplemental document \cite{proof}.
\end{proof}
Notably, the result in \textbf{Proposition} 2 shares a similar form with Corollary 1 in ~\cite{Iceberg} and can be viewed as a generalization of the latter. 

The aforementioned propositions establish a theoretical foundation for optimizing sensing performance based on the constellation kurtosis. The sidelobe amplitudes at $\tau = \frac{kN}{S}$, where $k \in \mathbb{Z}$, exhibit a positive correlation with $\mu_4$, whereas the impact of $\mu_4$ on the main peak is minimal. Consequently, reducing $\mu_4$ diminishes the Doppler-shift sidelobes along these delay slices, thereby enhancing velocity sensing performance. This is because $\mu_4$ weights the fourth-order terms in the expected squared AF, which are the dominant contributors to Doppler leakage away from the mainlobe on a fixed-delay cut. Hence, shaping the constellation to lower $\mu_4$ yields a cleaner Doppler response and improved velocity estimation.

Subsequently, we investigate the structural characteristics of the DAFT-s-AFDM AF by considering three distinct subcarrier allocation schemes:
\begin{itemize}
\item[($\mathrm{a}$)] $M=N$ and $S=1$ (full subcarrier occupancy, equivalent to a single-carrier waveform);
\item[($\mathrm{b}$)] $M<N$ and $S=1$ (L-FDMA);
\item[($\mathrm{c}$)] $M<N$ and $S=N/M$ (I-FDMA).
\end{itemize}

The term $\mathcal{T}_2$ of \eqref{eq:DAFT_AF} denotes the aggregate mutual AF components across all subcarriers. Based on the principle of energy conservation, as $N$ becomes large, the energy associated with $\mathcal{T}_2$ constitutes a dominant fraction of $\mathbb{E} \left\{ |\chi(\tau, \nu)|^2 \right\}$, thereby forming the primary structural sidelobes. Consequently, the structural properties of the AF are elucidated by analyzing $\mathcal{T}_2$, leading to the following proposition.

\begin{proposition}
\label{prop:3}
Consider the DPAF. For waveform configuration ($\mathrm{a}$), the main sidelobes form a flat constant pedestal of $N$. For waveform configuration ($\mathrm{b}$), the main sidelobes consist of a weighted sum of overlapping sinc functions centered around the $\phi$ axis. For waveform configuration ($\mathrm{c}$), the main sidelobes exhibit a comb-like structure with linear delay-Doppler coupling.
\end{proposition}
\begin{proof}
See \textbf{App.}~D in the supplemental document \cite{proof}.
\end{proof}

In contrast to $\mathcal{T}_2$, $\mathcal{T}_1$ primarily manifests as peaks in the AF, leading to the following proposition.

\begin{proposition}
\label{prop:4}
Within the defined domain, all coordinates $(\tau^*, \nu^*)$ of the local maxima of $\mathcal{T}_1$ are given by
\begin{equation}
\left( \tau^*, \nu^* \right) = \left( \frac{mN}{S}, \quad \frac{2N^2 c_1 m}{S} - kN \right), \quad m, k \in \mathbb{Z}.
\end{equation}
\end{proposition}
\begin{proof}
See \textbf{App.}~E in the supplemental document \cite{proof}.
\end{proof}

The aforementioned propositions offer valuable insights into the overall structure of the DAFT-s-AFDM AF. However, for parameter configurations ($\mathrm{b}$) and ($\mathrm{c}$), the resulting AF structure exhibits stronger delay-Doppler coupling and is more involved than that of configuration ($\mathrm{a}$). Therefore, simulation results are presented in Sec.~\ref{sec:results} to validate the theoretical analysis and offer additional insights into the AF behavior under these more complex settings.

% =========================================================================
% IV. FORMULATION OF COMMUNICATION PERFORMANCE VIA EFFECTIVE THROUGHPUT
% =========================================================================
\section{Effective Throughput Analysis for the DAFT-s-AFDM Waveform}
\label{sec:comm_performance}

To comprehensively evaluate the dual-functional performance of the waveform, this section introduces the effective throughput as the performance metric for the communication system.

\subsection{Effective Throughput}
\label{subsec:throughput_formulation}

The effective throughput metric, denoted as $\mathcal{I}(\mathbf{P}, \text{SNR})$ is formulated as the product of the source entropy and the probability of correct block transmission, yielding
\begin{equation}
\mathcal{I}(\mathbf{P}, \text{SNR}) = H(\mathbf{P}) \cdot (1 - P_b(\mathbf{P}, \text{SNR})),
\label{eq:throughput_main}
\end{equation}
where $\mathbf{P} = \{P(x)\}_{x \in \mathcal{X}}$ is the constellation's PMF, $\text{SNR}$ is the signal-to-noise ratio (SNR) of the AWGN channel, $H(\mathbf{P}) = -\sum_{i} p_i \log_2 p_i$ denotes the constellation entropy in bits/symbol, and $P_b(\mathbf{P}, \text{SNR})$ represents the BER.

\subsection{Bit Error Rate}
\label{subsec:ber_linear}
Recalling the system model described in Section~\ref{sec:system_model}, the received signal vector $\mathbf{y}$ as
\begin{align}
\mathbf{y} &= \mathbf{H}\mathbf{U}\mathbf{x} + \mathbf{n} \nonumber \\
&\triangleq \mathbf{H}_{\text{eff}}\mathbf{x} + \mathbf{n},
\end{align}
where $\mathbf{H}_{\text{eff}} = \mathbf{H}\mathbf{U} \in \mathbb{C}^{N \times M}$ denotes the effective channel matrix, and $\mathbf{n} \sim \mathcal{CN}(\mathbf{0}, N_0 \mathbf{I}_N)$ represents AWGN. The receiver employs a linear equalizer $\mathbf{W} \in \mathbb{C}^{M \times N}$ to process $\mathbf{y}$, producing the estimated signal vector
\begin{equation}
\mathbf{r} = \mathbf{W}\mathbf{y} = \mathbf{W}\mathbf{H}_{\text{eff}}\mathbf{x} + \mathbf{W}\mathbf{n}.
\end{equation}
The effective modulation matrix is defined as $\mathbf{G} = \mathbf{W}\mathbf{H}_{\text{eff}}$, and the post-equalization noise vector is denoted by $\tilde{\mathbf{n}} = \mathbf{W}\mathbf{n}$. For the $k$-th symbol, $r_k$ can be written as
\begin{equation}
\label{eq:rk_orign}
r_k = G_{kk} x_k + \sum_{l \neq k} G_{kl} x_l + \tilde{n}_k,
\end{equation}
where $G_{kk}$ is the diagonal element of $\mathbf{G}$, which represents the equivalent channel gain, and the off-diagnoal elements $G_{kl}$, which captures the interference from the $l$-th transmitted symbol to the $k$-th received symbol. By invoking the central limit theorem, the residual interference $\sum_{l \neq k} G_{kl} x_l$ is modeled as Gaussian noise $\tilde{n}_k$, hence \eqref{eq:rk_orign} can be rewritten as
\begin{equation}
\label{eq:rk}
r_k = \alpha_k x_k + z_k,
\end{equation}
where $\alpha_k = G_{kk}$ and $z_k = \sum_{l \neq k} G_{kl} x_l + \tilde{n}_k$ denotes the aggregate equivalent noise. Since both $\mathbf{x}$ and $\tilde{\mathbf{n}}$ are zero-mean, $z_k$ is also zero-mean. Its variance $\sigma_k^2 \triangleq \mathbb{E}[|z_k|^2]$ consists of two components:
\begin{align}
\sigma_{\text{noise}, k}^2 &= \left[ \mathbf{W} (N_0 \mathbf{I}_N) \mathbf{W}^H \right]_{kk} = N_0 \|\mathbf{w}_k\|^2, \\
\sigma_{\text{int}, k}^2 &= \sum_{l \neq k} |G_{kl}|^2 E_s,
\end{align}
where $\mathbf{w}_k$ denotes the $k$-th row of $\mathbf{W}$ and $E_s$ is the symbol energy of $x$. Therefore, the overall noise variance is given by
\begin{equation}
\sigma_k^2 = N_0 \|\mathbf{w}_k\|^2 + E_s \sum_{l \neq k} |G_{kl}|^2.
\end{equation}

As the original noise $\mathbf{n}$ is circularly symmetric and the linear transformation $\mathbf{W}$ preserves Gaussianity and circular symmetry, $z_k$ follows a circularly symmetric complex Gaussian (CSCG) distribution with zero mean, i.e., $z_k \sim \mathcal{CN}(0, \sigma_k^2)$. At the receiver, the maximum a posteriori (MAP) detector can be expressed as
\begin{equation}
\label{eq:MAP_detector}
\hat{x}_k = \arg \max_{x_m \in \mathcal{X}} P(x_m | r_k).
\end{equation}
By invoking the Bayes' theorem, \eqref{eq:MAP_detector} can be simplified as
\begin{align}
\label{eq:x_k_likelihood}
\hat{x}_k &= \arg \max_{x_m \in \mathcal{X}} \left\{ p(r_k | x_m) P(x_m) \right\} \nonumber \\
&= \arg \max_{x_m \in \mathcal{X}} \left\{ \frac{1}{\pi \sigma_k^2} \exp\left( -\frac{|r_k - \alpha_k x_m|^2}{\sigma_k^2} \right) P(x_m) \right\}.
\end{align}
It can be readily shown that the optimization problem of \eqref{eq:x_k_likelihood} can be reformulated as
\begin{equation}
\hat{x}_k = \arg \min_{x_m \in \mathcal{X}} \left( \frac{|r_k - \alpha_k x_m|^2}{\sigma_k^2} - \ln P(x_m) \right).
\end{equation}
Consider the case where the transmitted symbol $x_i$ is erroneously detected as $x_j$. An error occurs when the metric for $x_j$ is smaller than that for $x_i$, i.e.,
\begin{equation}
\label{eq:rk_uk}
\frac{|r_k - \alpha_k x_j|^2}{\sigma_k^2} - \ln P(x_j) < \frac{|r_k - \alpha_k x_i|^2}{\sigma_k^2} - \ln P(x_i).
\end{equation}
Based on \eqref{eq:rk} and define the effective distance vector $d_{ij} \triangleq \alpha_k (x_i - x_j)$, \eqref{eq:rk_uk} is equivalent to
\begin{equation}
|\alpha_k x_i + z_k - \alpha_k x_j|^2 - |z_k|^2 < \sigma_k^2 \left( \ln P(x_j) - \ln P(x_i) \right).
\end{equation}
Applying the geometric identity $|d_{ij} + z_k|^2 = |d_{ij}|^2 + |z_k|^2 + 2 \operatorname{Re}\{ d_{ij}^* z_k \}$, the expression simplifies to
\begin{equation}
2 \operatorname{Re}\{ d_{ij}^* z_k \} < -\sigma_k^2 \ln \frac{P(x_i)}{P(x_j)} - |d_{ij}|^2 .
\end{equation}
Define the decision variable $Y \triangleq 2 \operatorname{Re}\{ d_{ij}^* z_k \}$. As $z_k$ is CSCG noise, the projected component $Y$ follows a zero-mean real Gaussian distribution with variance $\operatorname{Var}(Y) = 4 |d_{ij}|^2 \cdot \frac{\sigma_k^2}{2} = 2 |d_{ij}|^2 \sigma_k^2$. The pairwise error probability (PEP) is thus the probability that $Y$ falls below the threshold $T = -|d_{ij}|^2 - \sigma_k^2 \ln \frac{P(x_i)}{P(x_j)}$, given by
\begin{equation}
P(Y < T) = Q\left( \frac{|d_{ij}|^2 + \sigma_k^2 \ln \frac{P(x_i)}{P(x_j)}}{|d_{ij}| \sqrt{2 \sigma_k^2}} \right).
\end{equation}
Applying the union bound, the BER is approximated as
\begin{equation}
\label{eq:BER_1}
\begin{split}
P_b  \approx  &  \frac{1}{\log_2{|\mathcal{X}|}} \sum_{i=1}^{|\mathcal{X}|} P(x_i) \sum_{j \neq i}^{|\mathcal{X}|}  d_H(x_i, x_j) \\ & \quad \quad \times Q\left( \frac{|\alpha_k (x_i - x_j)|^2 + \sigma_k^2 \ln \left( \frac{P(x_i)}{P(x_j)} \right)}{\sqrt{2 \sigma_k^2} \, |\alpha_k (x_i - x_j)|} \right),
\end{split}
\end{equation}
where $d_H(x_i, x_j) = \sum_{k=1}^{m} b_{i,k} \oplus b_{j,k}$ denotes the Hamming distance, with $b_{i,k}$ denoting the $k$-th bit of the binary label for symbol $x_i$, $m = \log_2 |\mathcal{X}|$, and $\oplus$ representing the modulo-2 addition.

Let the $|\mathcal{X}|$ constellation points be partitioned into $K$ energy rings, denoted by the set $\mathcal{E} = \{E_1, E_2, \dots, E_K\}$. The subset of symbols located on the $r$-th energy ring is defined as
\begin{equation}
\mathcal{X}_r = \{x_m \in \mathcal{X} \mid |x_m|^2 = E_r\}.
\end{equation}
Since the generalized Maxwell-Boltzmann distribution depends only on $|x_i|^2$ and $|x_i|^4$, the probability of any symbol $x_m \in \mathcal{X}_r$ is given by $P(E_r)$. At high SNR, detection errors are dominated by nearest neighbors. Let $d_{\min}$ denote the minimum Euclidean distance of the base constellation $\mathcal{X}$. The topological Hamming-weight matrix entry $C(E_r, E_s)$, which aggregates the total Hamming distance from all symbols on ring $E_r$ to their nearest neighbors on ring $E_s$, is defined as
\begin{equation}
C(E_r, E_s) = \sum_{x_i \in \mathcal{X}_r} \sum_{x_j \in \mathcal{N}(x_i) \cap \mathcal{X}_s} d_H(x_i, x_j),
\end{equation}
where $\mathcal{N}(x_i) = \{x \in \mathcal{X} \mid |x_i - x| = d_{\min}\}$ denotes the set of nearest neighbors for symbol $x_i$. By substituting $P(x_i) = P(E_r)$, $P(x_j) = P(E_s)$, and $|x_i - x_j| = d_{\min}$ into \eqref{eq:BER_1}, the symbol-wise double summation is consolidated into a ring-wise summation. The BER is thus approximated as
\begin{equation}
\begin{split}
P_b \approx & \frac{1}{\log_2 |\mathcal{X}|} \sum_{r=1}^K P(E_r) \sum_{s=1}^K C(E_r, E_s)
\\
& \quad \quad \times Q\left( \frac{|\alpha_k|^2 d_{\min}^2 + \sigma_k^2 \ln \left( \frac{P(E_r)}{P(E_s)} \right)}{|\alpha_k| d_{\min} \sqrt{2\sigma_k^2}} \right).
\end{split}
\end{equation}
Under this formulation, the computational complexity of evaluating the theoretical BER is reduced from $\mathcal{O}(|\mathcal{X}|^2)$ to $\mathcal{O}(K^2)$. 

% =========================================================================
% V. PROPOSED PCS WAVEFORM OPTIMIZATION
% =========================================================================
\section{Proposed PCS-based ISAC Waveform Design}
\label{sec:optimization}

In this section, a dual-functional waveform design is proposed to jointly improve the  AF characteristics and communication throughput.

\subsection{Problem Formulation}
The fundamental trade-off between communication throughput and sensing accuracy is addressed by formulating the following multi-objective optimization problem (MOOP)
\begin{subequations}
\label{eq:general_opt}
\begin{align}
\min_{\mathbf{P}}& \quad \left[ -\mathcal{I}(\mathbf{P}, \text{SNR}), \mu_4 \right] \\
\text{s.t.} & \quad \sum_{x \in \mathcal{X}} p(x) = 1, \quad p(x) \ge 0, \label{eq:C1}
\\
& \quad \mathbb{E}[|x|^2] = \sum_{x \in \mathcal{X}} p(x) |x|^2 = \sigma_x^2,
\label{eq:C2}
\end{align}
\end{subequations}
where $\mathbf{P}$ is the PMF to be optimized, $\mu_4$ is the fourth-order moment representing sensing performance, and $\mathcal{I}(\mathbf{P}, \text{SNR})$ is the effective throughput metric of \eqref{eq:throughput_main}. Given the typically high dimensionality of the constellation's probability distribution, a direct search for the optimal
$\mathbf{P}$ over the entire probability distribution space is computationally intractable. Therefore, to explore the trade-off between communication and sensing, we begin by formulating a multi-objective optimization problem. The objective is to simultaneously maximize the constellation entropy $H(\mathbf{P})$ (for communication) and the fourth-order moment $\mathbb{E}[|x|^4]$ (for sensing), subject to basic probability and power constraints, yielding
\begin{subequations}
\label{eq:pmf_multi_obj_origin}
\begin{align}
\min_{\mathbf{P}} & \quad \left[ -H(\mathbf{P}), \; \mu_4 \right] \\
\text{s.t.} & \quad \eqref{eq:C1}, \eqref{eq:C2}.
\end{align}
\end{subequations}

This problem can be solved using the weighted sum and the Lagrange multipliers algorithms. By solving this problem using non-negative weights $\omega_1, \omega_2$ (with $\omega_1 + \omega_2 = 1$), we find that any Pareto-optimal PMF must conform to a generalized Maxwell-Boltzmann (MB) distribution
\begin{equation}
\label{eq:pmf_form}
p(x; \lambda_1, \lambda_2) = \frac{1}{Z(\lambda_1, \lambda_2)} \exp\left( \lambda_1 |x|^2 + \lambda_2 |x|^4 \right),
\end{equation}
where $Z(\lambda_1, \lambda_2) = \sum_{x \in \mathcal{X}} \exp(\lambda_1 |x|^2 + \lambda_2 |x|^4)$. The parameter $\lambda_2 = \frac{\omega_2 \ln 2}{\omega_1}$ is a trade-off coefficient that governs the sensing performance, while $\lambda_1$ is implicitly fixed by the average power constraint, which is given by
\begin{equation}
\label{eq:constr_power}
\sum_{x \in \mathcal{X}} p(x; \lambda_1, \lambda_2) |x|^2 = \sigma_x^2.
\end{equation}
Therefore, for given weights $\omega_1, \omega_2$, $\lambda_2$ is directly determined.

\eqref{eq:pmf_form} provides a principled and computationally efficient foundation for our waveform design. Therefore, the original problem can be further reformulated as
\begin{subequations}
\label{eq:pmf_multi_obj}
\begin{align}
\min_{\mathbf{P}} & \quad \left[  -\mathcal{I}(\mathbf{P}, \text{SNR}), \; \mu_4 \right] \\
\text{s.t.} & \quad  p(x; \lambda_1, \lambda_2) = \frac{\exp \left( \lambda_1 |x|^2 + \lambda_2 |x|^4 \right)}{Z(\lambda_1, \lambda_2)} \label{eq:Probelm_PMF}
% \quad  \mathbf{P}(\lambda_1, \lambda_2) \text{ is given by \eqref{eq:pmf_form}}
\end{align}
\end{subequations}

The PMF $\mathbf{P}$ is generated using the parametric form from~\eqref{eq:pmf_form}, which allows optimization over the parameter $\lambda_1$ and $\lambda_2$.
To solve \eqref{eq:pmf_multi_obj} and trace the Pareto-optimal front, we employ the widely used weighted sum method. To transform the problem into a single-objective optimization problem. The weighted linear combination of objectives can be expressed as
\begin{equation}
\label{eq:weighted_sum_obj}
J(\mathbf{P}, \omega) = -\omega \cdot \mathcal{I}(\mathbf{P}, \text{SNR}) + (1-\omega) \cdot  \mu_4,
\end{equation}
where $\omega \in [0, 1]$ is the weighting factor that balances the two objectives. Finally, \eqref{eq:pmf_multi_obj} is reformulated as
\begin{subequations}
\label{eq:pmf_multi_obj_final}
\begin{align}
\min_{\mathbf{P}} & \quad J(\mathbf{P}, \omega)\\
\text{s.t.} & \quad \mathbf{P}(\lambda_1, \lambda_2) \text{ is given by \eqref{eq:Probelm_PMF}}
\end{align}
\end{subequations}

\subsection{ISAC Waveform Optimization}
\label{subsec:linear_opt}

The optimization objective is to find the PMF $\mathbf{P}$ that maximizes the scalarized utility function $J(\mathbf{P}, \omega)$, representing the weighted sum of throughput and sensing performance. Based on the derivation in \eqref{eq:pmf_form}, the optimal PMF which satisfying the maximum entropy principle is strictly parameterized by two Lagrange multipliers, $\lambda_1$ and $\lambda_2$. These two multipliers are which correspond to the second-order moment and the fourth-order moment, respectively. Consequently, the problem of optimizing the high-dimensional vector $\mathbf{P}$ is transformed into finding the optimal pair of parameters $(\lambda_1, \lambda_2)$.

To solve this problem, we propose a robust two-stage search procedure, detailed in \textbf{Algorithm}~\ref{alg:linear_opt}, to find the optimal operating point. The hybrid strategy first performs a coarse grid search over the trade-off weight $\omega$ and the average power $\sigma_x^2$ to globally identify the most promising operating region. The search grids $\mathcal{W}_{\text{grid}} \subset [0,1]$ and $\mathcal{P}_{\text{grid}} \subset \mathbb{R}_+$ are predefined candidate sets for the trade-off weight $\omega$ and the average power $\sigma_x^2$, respectively. In this stage, for each grid point $(\omega, \sigma_x^2)$, the second Lagrange multiplier as $\lambda_2 = \frac{(1-\omega) \ln 2}{\omega}$, then we calculate the first Lagrange multiplier $\lambda_1$ to satisfy the power constraint in \eqref{eq:constr_power}. Construct the corresponding PMF as $\mathbf{P} = p(x; \lambda_1, \lambda_2)$ and evaluate the composite objective $J(\mathbf{P}, \omega)$ using \eqref{eq:weighted_sum_obj}. The best pair $(\omega^*, \sigma_x^{2*})$ from this search can be invoked to initialize a high-precision local refinement with the Nelder-Mead simplex method \cite{fminsearch}, yielding a derivative-free direct search algorithm for unconstrained nonlinear optimization that iteratively updates a simplex using reflection, expansion, contraction, and shrink operations. In the refinement stage, we optimize over $(\omega, \sigma_x^2)$ to obtain $(\omega_{\text{opt}}, \sigma_{x,\text{opt}}^2)$, compute $\lambda_{2,\text{opt}} = \frac{(1-\omega_{\text{opt}})\ln 2}{\omega_{\text{opt}}}$, determine $\lambda_{1,\text{opt}}$ by enforcing the power constraint with $(\omega_{\text{opt}}, \sigma_{x,\text{opt}}^2)$, and finally construct the PMF as $\mathbf{P}^* = p(x; \lambda_{1,\text{opt}}, \lambda_{2,\text{opt}})$. This approach efficiently finds the globally optimal PMF that minimizes the objective value under the given channel configuration.

\begin{algorithm}[t!]
\SetAlgoLined
\caption{Hybrid Grid-Refinement Optimizer}
\label{alg:linear_opt}
% Input and Output
\KwIn{Constellation set $\mathcal{X}$, SNR, search grids $\mathcal{W}_{\text{grid}}$, $\mathcal{P}_{\text{grid}}$}
\KwOut{The single best PMF $\mathbf{P}^*$}
% Initialization
\BlankLine
$J_{\text{best}} \leftarrow -\infty$\;
$(\omega^*, \sigma_x^{2*}) \leftarrow (\emptyset, \emptyset)$\;
% \tcp{Stage 1: Global Grid Search}
\ForEach{$\omega \in \mathcal{W}_{\text{grid}}$}{
\ForEach{$\sigma_x^2 \in \mathcal{P}_{\text{grid}}$}{
Compute multipliers $(\lambda_1, \lambda_2)$\;
Construct PMF $\mathbf{P}$ and compute $J(\mathbf{P}, \omega)$\;
\If{$J > J_{\text{best}}$}{
  $J_{\text{best}} \leftarrow J$; \ $(\omega^*, \sigma_x^{2*}) \leftarrow (\omega, \sigma_x^2)$\;
}
}
}

\If{$(\omega^*, \sigma_x^{2*}) \neq (\text{null}, \text{null})$}{
Let initial point be $\mathbf{v}_0 = (\omega^*, \sigma_x^{2*})$\;
$(\omega_{\text{opt}}, \sigma_{x, \text{opt}}^2) \leftarrow \text{Nelder-Mead}(J, \mathbf{v}_0)$\;
Compute $\lambda_{2, \text{opt}}$ and $\lambda_{1, \text{opt}}$\;
$\mathbf{P}^* \leftarrow p(x; \lambda_{1, \text{opt}}, \lambda_{2, \text{opt}})$\;
}
\Return{$\mathbf{P}^*$}\;
\end{algorithm}

\subsection{Complexity Analysis}
\label{subsec:complexity}
This subsection analyzes the computational complexity of \textbf{Algorithm}~\ref{alg:linear_opt}. Denote the grid sizes as $|\mathcal{W}_{\text{grid}}|$ and $|\mathcal{P}_{\text{grid}}|$, let $K$ be the number of independent energy rings. For each grid point $(\omega,\sigma_x^2)$, the algorithm computes $(\lambda_1,\lambda_2)$, constructs the PMF, and evaluates $J(\mathbf{P},\omega)$. The dominant cost is the evaluation of the BER approximation in Sec.~\ref{subsec:ber_linear}. Using the proposed energy-ring grouping framework, each evaluation has complexity $\mathcal{O}(K^2)$. Therefore, the overall complexity of the grid-search stage is on the order of $\mathcal{O}\big(|\mathcal{W}_{\text{grid}}|\,|\mathcal{P}_{\text{grid}}|\,K^2\big)$, where the cost of solving $\lambda_1$ to meet \eqref{eq:constr_power} is absorbed into the constant factor and can be efficiently computed by a one-dimensional root-finding method with a small number of iterations. The best grid point then initializes a Nelder-Mead refinement. Let $N_{\text{NM}}$ denote the number of objective evaluations required by this routine, and it should be noted that each evaluation has the same dominant cost. Hence, the refinement stage has complexity $\mathcal{O}\big(N_{\text{NM}}\,K^2\big)$, and the total computational complexity is on the order of
\begin{equation}
\mathcal{O}\Big(\big(|\mathcal{W}_{\text{grid}}|\,|\mathcal{P}_{\text{grid}}| + N_{\text{NM}}\big)\,K^2\Big).
\end{equation}
which is substantially lower than classical PCS baselines requiring repeated Monte Carlo (MC) mutual information estimation, enabling online, real-time transceiver optimization.

\begin{figure}[t!]
\centering

\begin{subfigure}{\linewidth}
\centering
\includegraphics[width=0.48\linewidth]{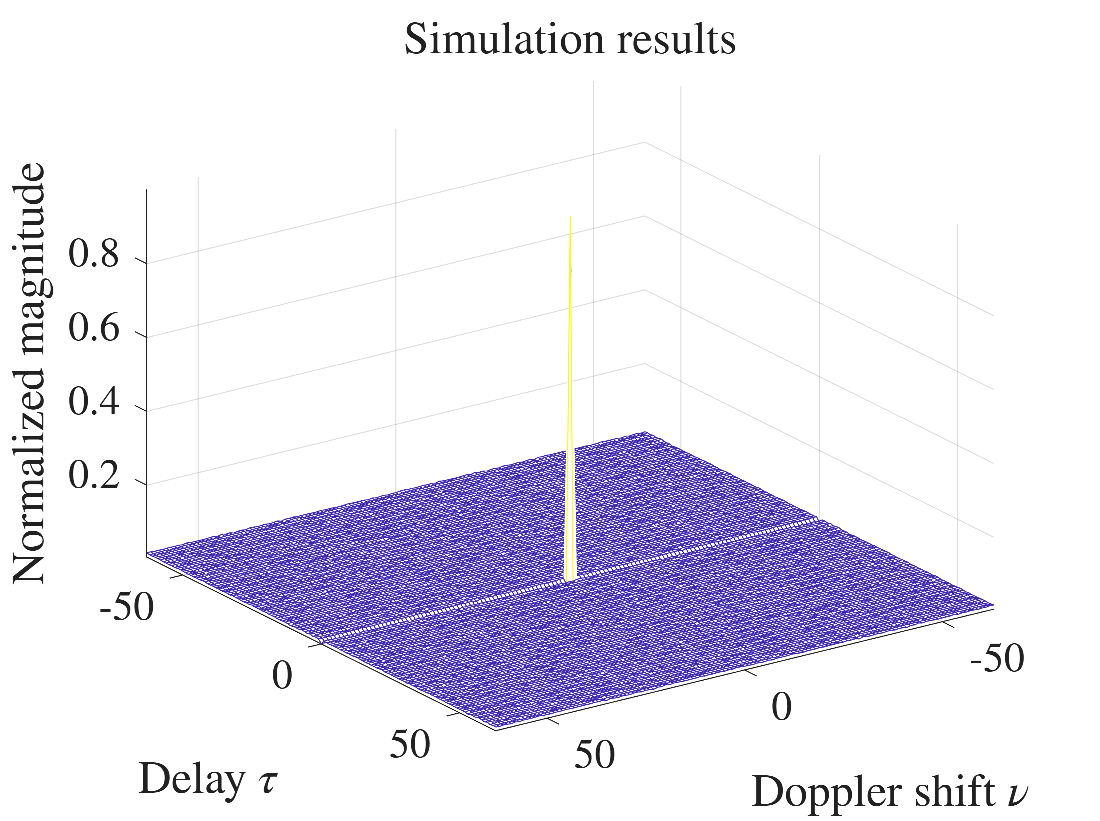}
\hfill
\includegraphics[width=0.48\linewidth]{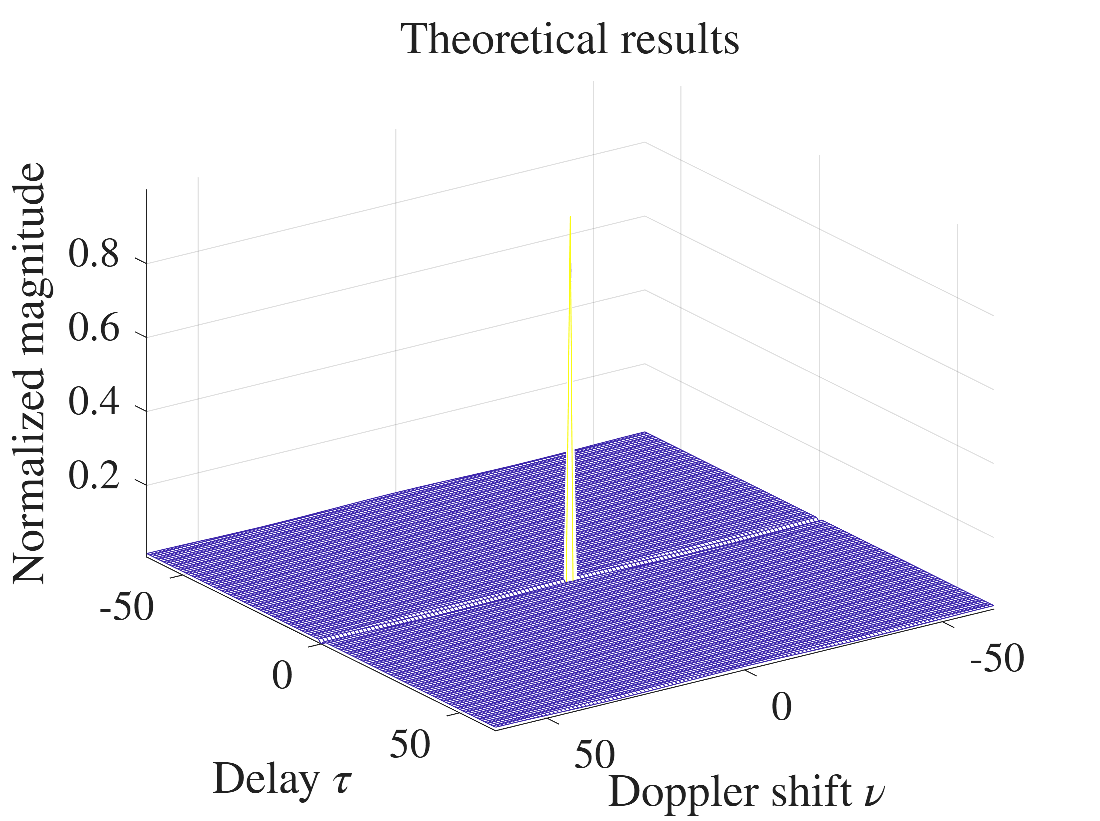}
\caption{Subcarrier configuration $\mathrm{(a)}$.}
\label{fig:AF_a}
\end{subfigure}

\vspace{1em}

\begin{subfigure}{\linewidth}
\centering
\includegraphics[width=0.48\linewidth]{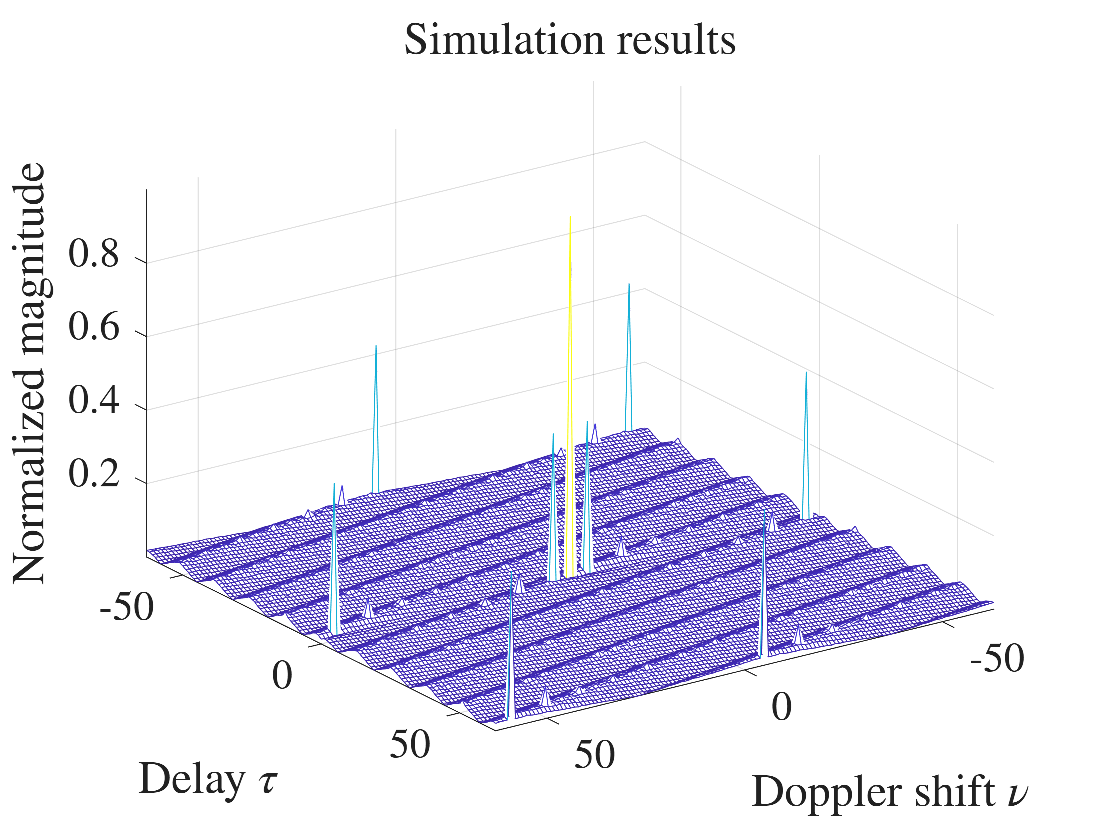}
\hfill
\includegraphics[width=0.48\linewidth]{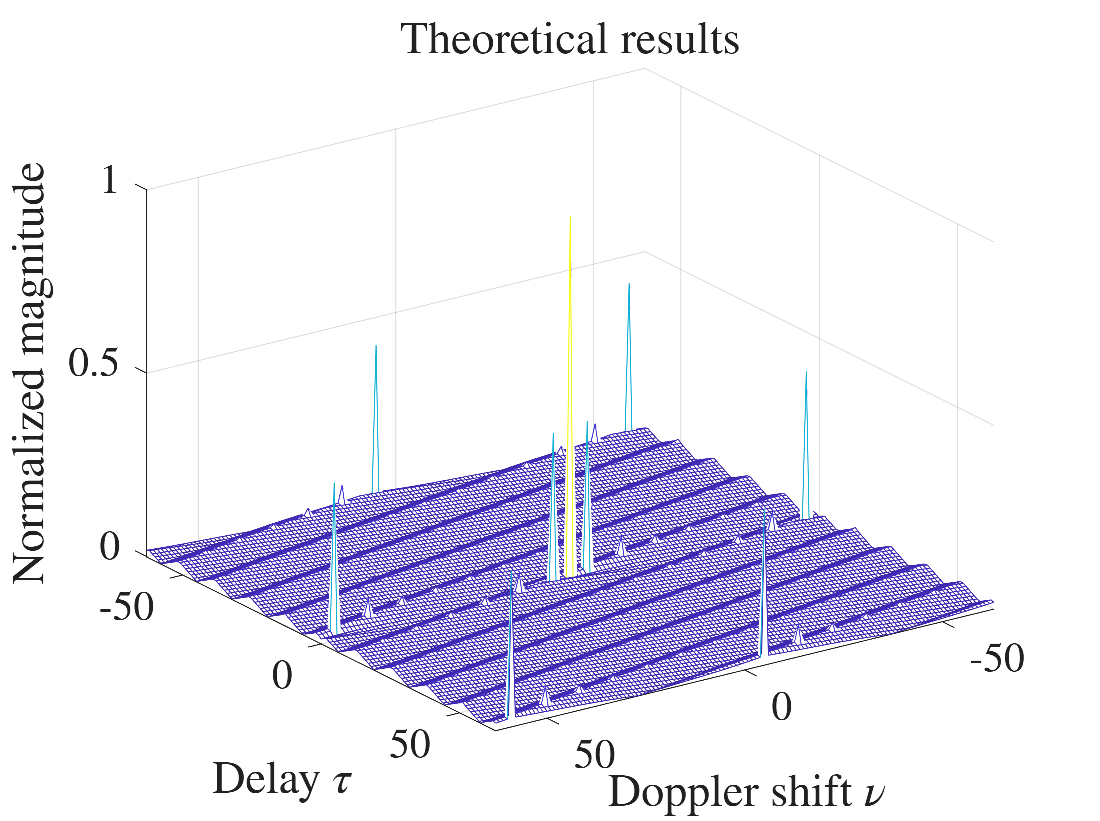}
\caption{Subcarrier configuration $\mathrm{(b)}$.}
\label{fig:AF_b}
\end{subfigure}

\vspace{1em}

\begin{subfigure}{\linewidth}
\centering
\includegraphics[width=0.48\linewidth]{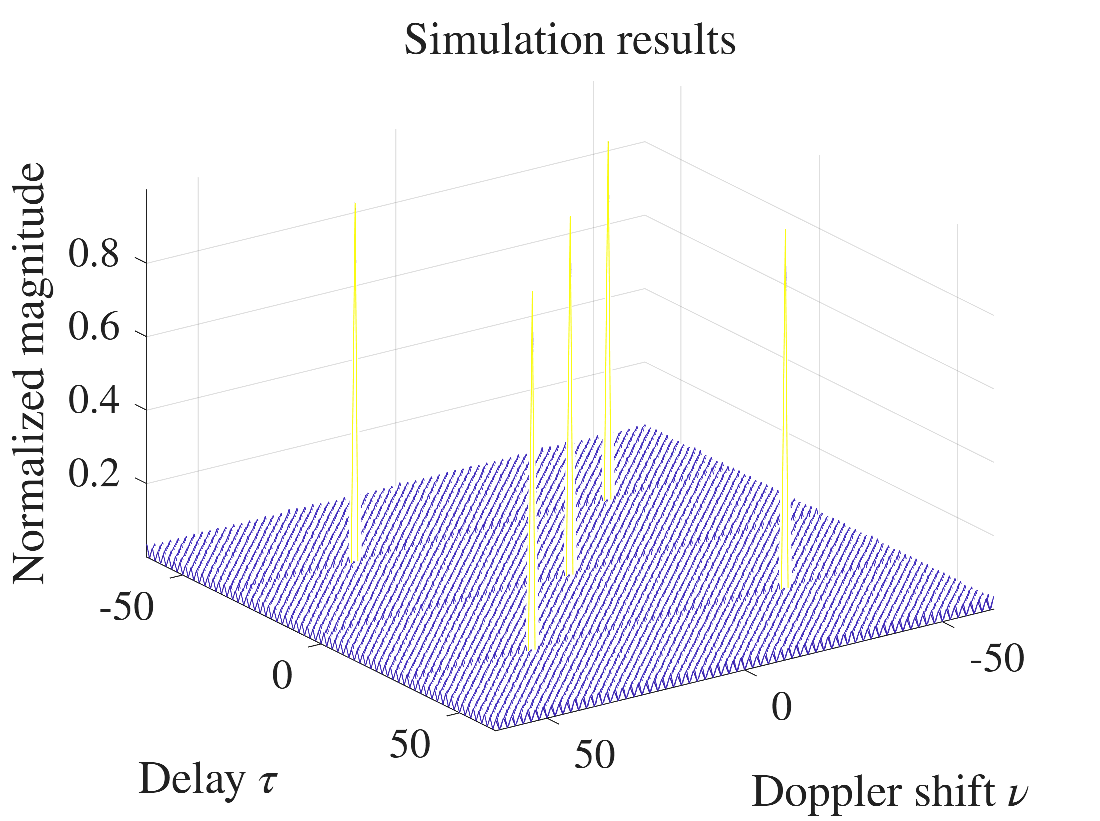}
\hfill
\includegraphics[width=0.48\linewidth]{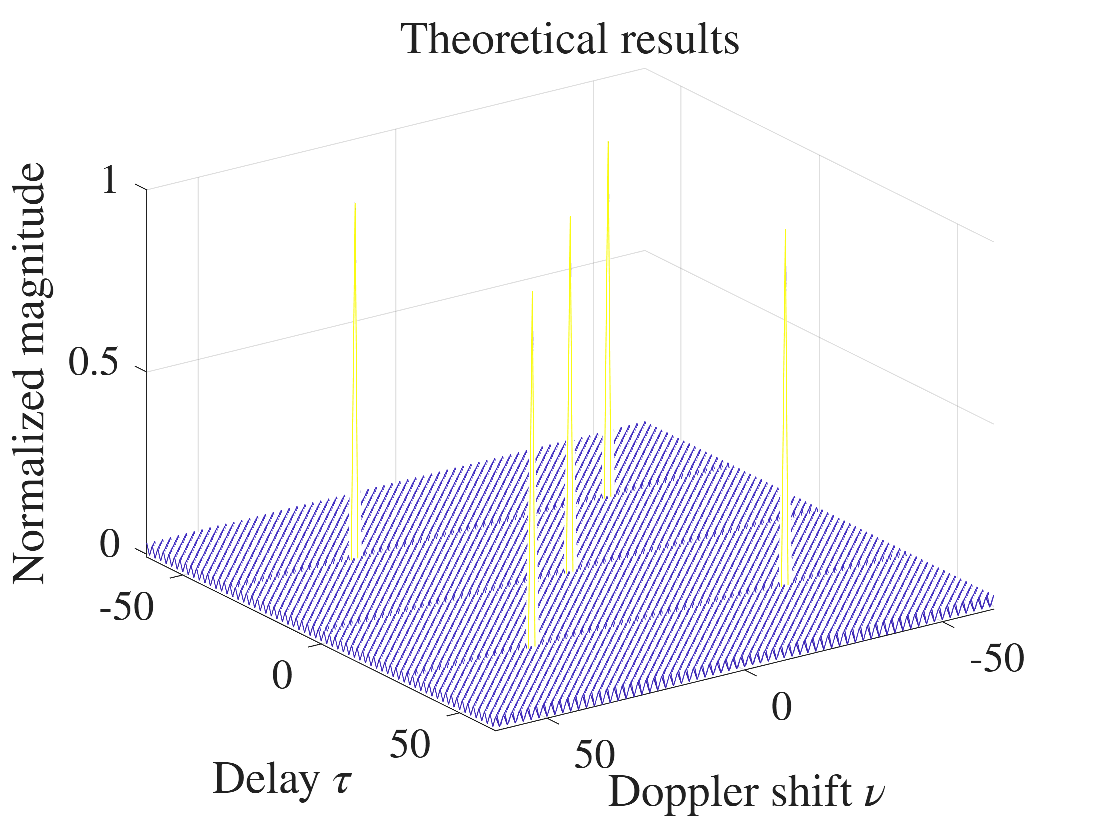}
\caption{Subcarrier configuration $\mathrm{(c)}$.}
\label{fig:AF_c}
\end{subfigure}

\caption{The DPAF for different subcarrier configurations.}
\label{fig:AF_abc}
\vspace{-6mm}
\end{figure}

\begin{figure}
\centering
\begin{subfigure}[t]{0.32\textwidth}
\centering
\includegraphics[width=\linewidth]{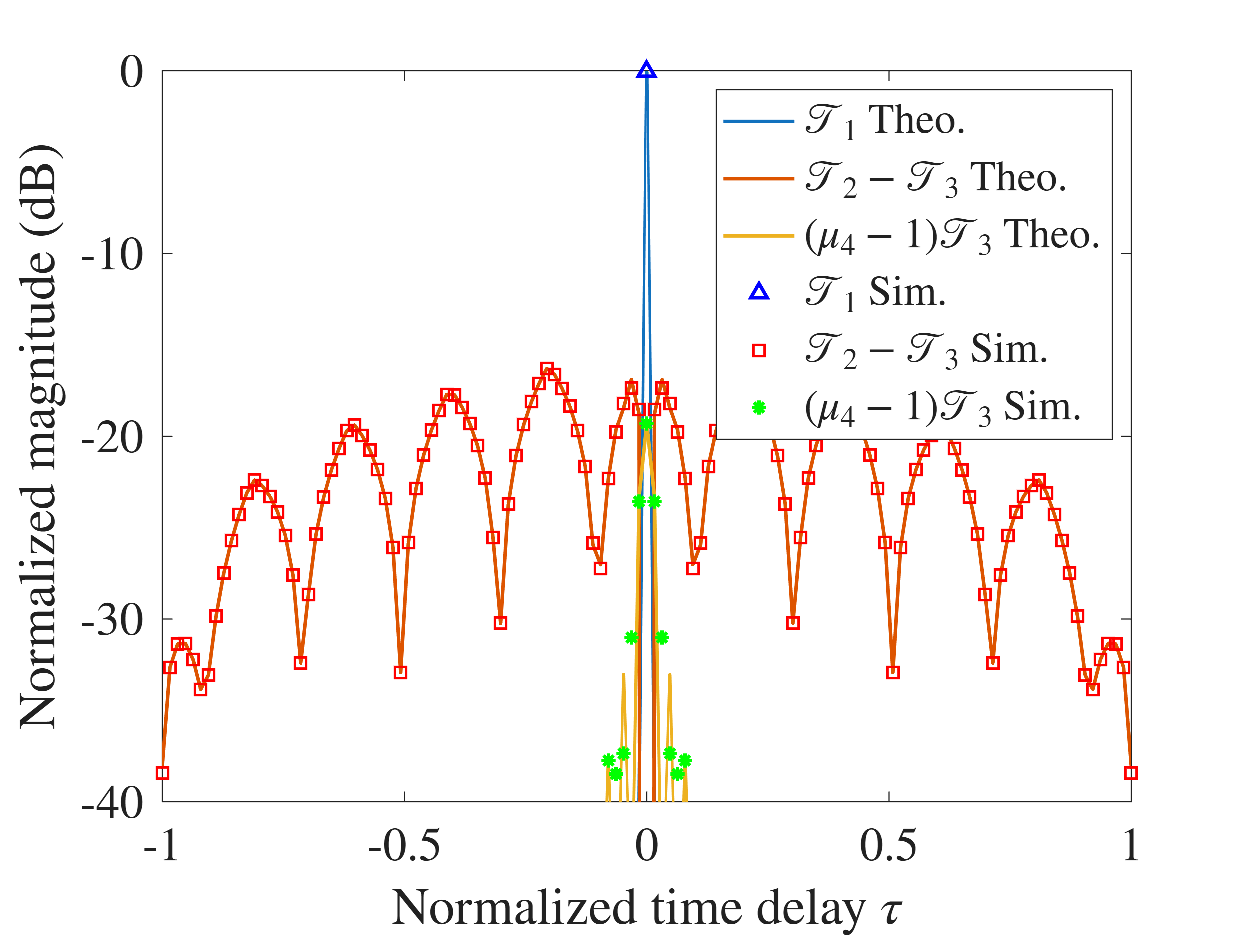}
\caption{Zero-doppler slice.}
\label{fig:af_decomp_range}
\end{subfigure}
\hfill  % 或者使用 \hspace{0.02\textwidth} 控制固定间距
\begin{subfigure}[t]{0.32\textwidth}
\centering
\includegraphics[width=\linewidth]{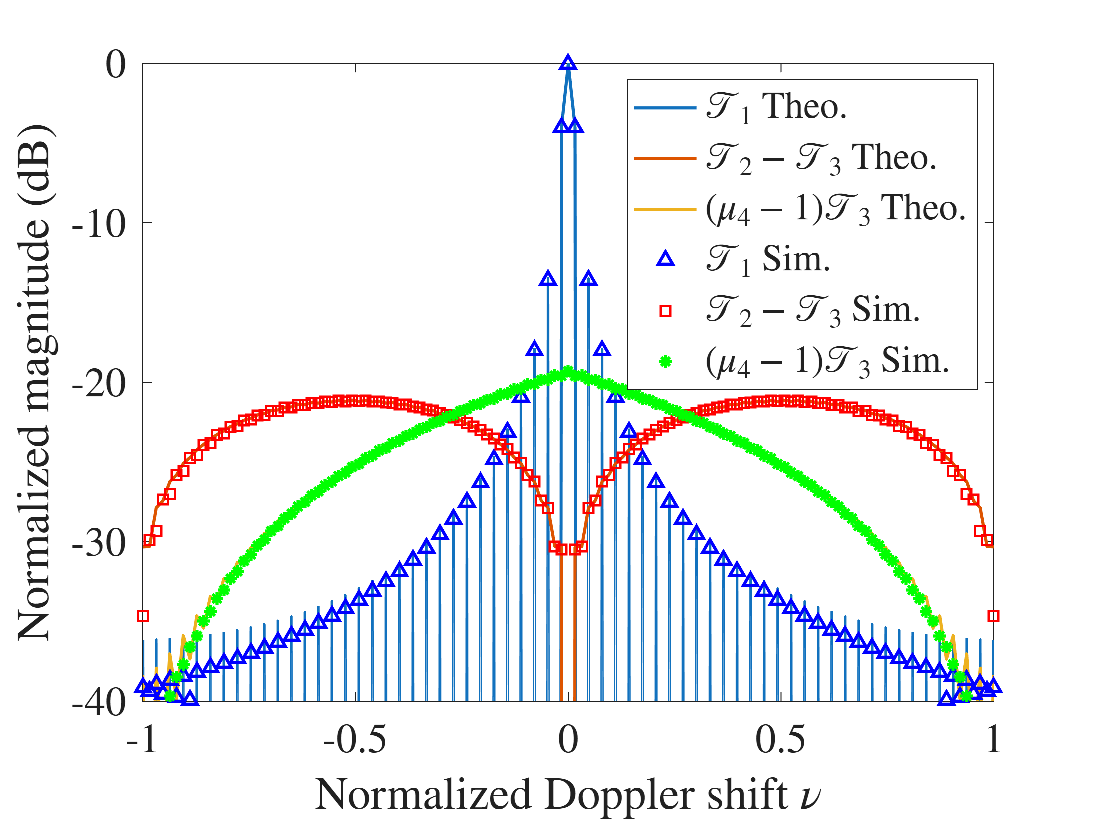}
\caption{Zero-delay slice.}
\label{fig:af_decomp_velocity}
\end{subfigure}

\caption{The DAAF slices for L-FDMA.}
\label{fig:af_decomp}
% \vspace{-3mm}
\end{figure}

% =========================================================================
% VI. NUMERICAL RESULTS
% =========================================================================
\section{Numerical Results}
\label{sec:results}

In this section, we present numerical results to validate our theoretical analyses and demonstrate the performance of the proposed PCS-based ISAC waveform design.

\subsection{Characterization of AF}
\label{subsec:pcs_af_validation}

Fig.~\ref{fig:AF_abc} depicts the DPAF for DAFT-s-AFDM under waveform configurations $\mathrm{(a)}$, $\mathrm{(b)}$, and $\mathrm{(c)}$, respectively. For $\mathrm{(a)}$, we set $N=M=64$; $N=64$ and $M=32$ are invoked for $\mathrm{(b)}$ and $\mathrm{(c)}$. Standard 64-QAM modulation is employed in all cases. It can be observed that the simulated AF align closely with the theoretical predictions, validating the derived expressions. For configuration $\mathrm{(a)}$, Fig.~\ref{fig:AF_a} reveals a pedestal of magnitude $1/N$ in $\mathbb{E} \left\{ |\chi(\tau, \nu)|^2 \right\}$, which is consistent with the observations of OFDM \cite{DPAF_TIT}. In this scenario, $\mathcal{T}_1$ exhibits a single peak within the defined domain, located at the origin. Under configuration $\mathrm{(b)}$, as shown in Fig.~\ref{fig:AF_b}, $\mathbb{E} \left\{ |\chi(\tau, \nu)|^2 \right\}$ features multiple slanted, blade-like ridges with delay-Doppler coupling, resembling the superposition of AFs from linear frequency modulation (LFM) signals with distinct delay shifts. This structure poses two challenges. First, deep valleys between ridges with varying delay-Doppler couplings may lead to substantial degradation for weak targets in multi-target sensing scenarios, as their main peaks could coincide with these valleys. Second, akin to LFM waveforms, DAFT-s-AFDM waveforms encounter difficulties in resolving multiple targets that exhibit delay-Doppler couplings inherent to the waveform. These limitations can be addressed through approaches such as CLEAN algorithm or V-shaped chirp modulation~\cite{CLEAN, FM_OFDM}. Compared to non-chirped modulations like OFDM and discrete Fourier transform-spread OFDM (DFT-s-OFDM), the chirped subcarriers in DAFT-s-AFDM under L-FDMA occupy the full signal bandwidth, rendering the delay resolution largely independent of $M$. This yields superior delay resolution given the subcarrier occupancy $M$, constituting a key advantage. Fig.~\ref{fig:AF_c} illustrates the AF under configuration $\mathrm{(c)}$ of I-FDMA, which manifests as a two-dimensional interleaved comb structure across the delay-Doppler plane, with multiple peaks within the defined domain. Although these peaks constrain the ambiguity-free sensing range, the zero-Doppler slice of the DAFT-s-AFDM  lacks spurious peaks, unlike DFT-s-OFDM. This preserves undistorted ambiguity-free distance at low relative target speeds.

\begin{figure}[t!]
\centering
\begin{subfigure}[b]{0.48\linewidth}
\centering
\includegraphics[width=\linewidth]{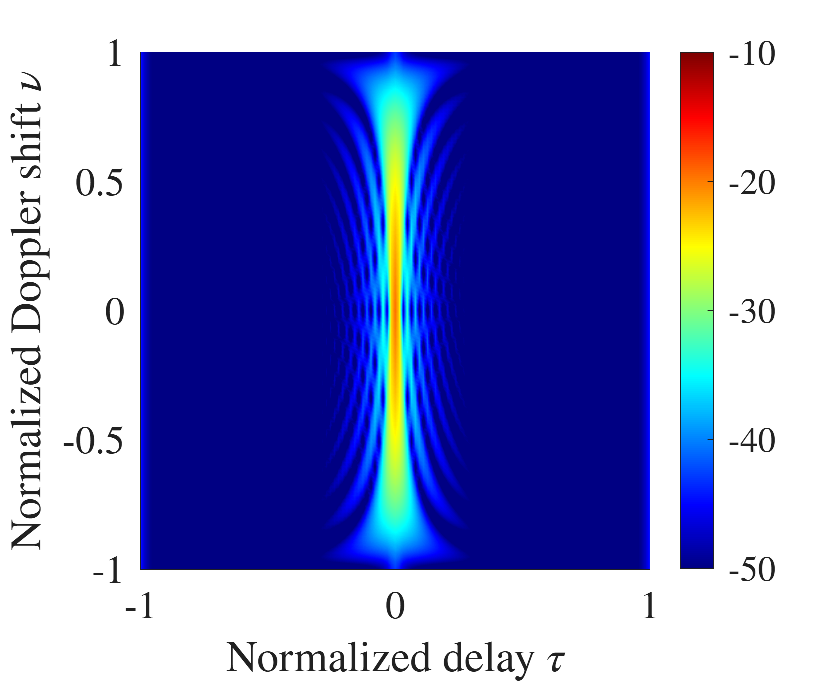}
\caption{$c_1 = 0$, $\Delta \lambda = 0$}
\label{fig:AF_T3_DFT}
\end{subfigure}
\hfill
\begin{subfigure}[b]{0.48\linewidth}
\centering
\includegraphics[width=\linewidth]{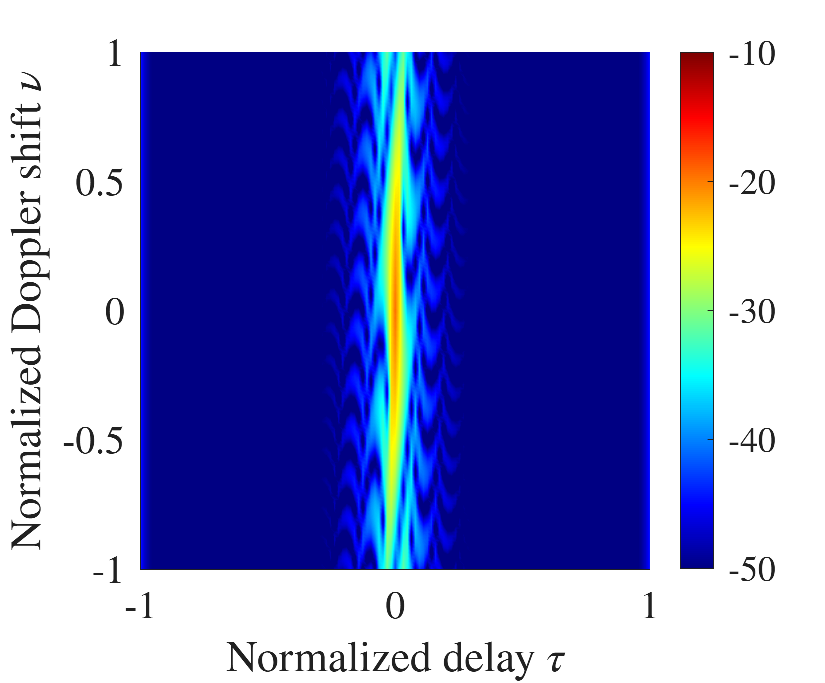}
\caption{$c_1 = 5/2N$, $\Delta \lambda = 0$}
\label{fig:AF_T3_DAFT}
\end{subfigure}

\vspace{1em} 

\begin{subfigure}[b]{0.48\linewidth}
\centering
\includegraphics[width=\linewidth]{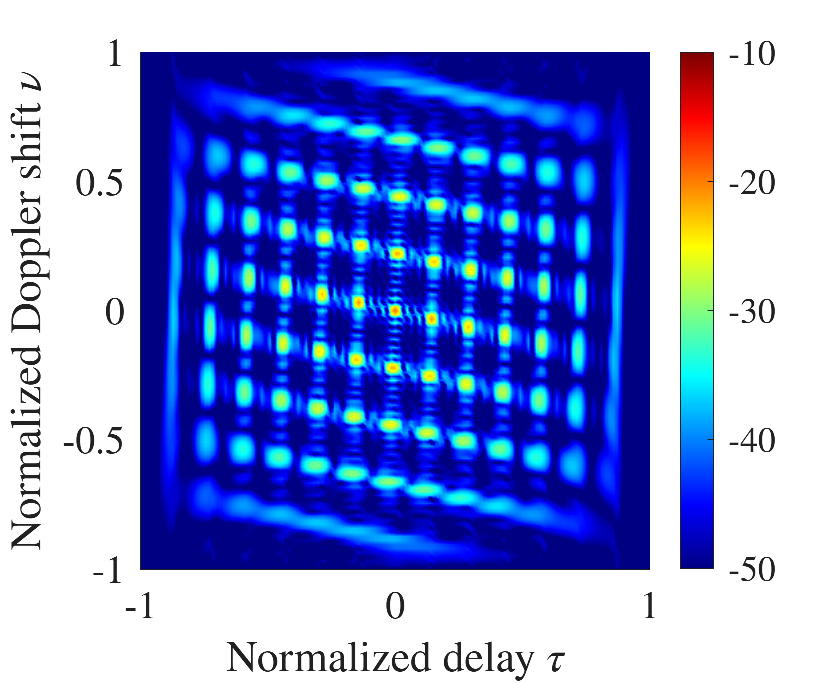}
\caption{$c_1 = 0$, $\Delta \lambda = \frac{\pi}{2}$.}
\label{fig:AF_T3_delta_lambda}
\end{subfigure}
\hfill
\begin{subfigure}[b]{0.48\linewidth}
\centering
\includegraphics[width=\linewidth]{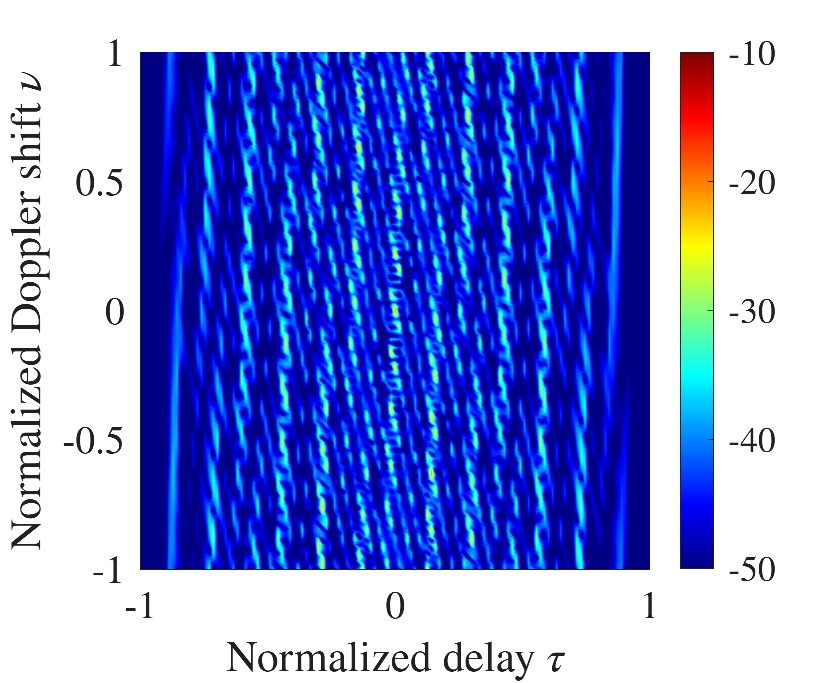}
\caption{$c_1 = 5/2N$, $\Delta \lambda = \frac{\pi}{2}$.}
\label{fig:AF_T3_DFT_delta_lambda}
\end{subfigure}

\caption{$\mathcal{T}_3$ of the AF for different parameter settings.}
\label{fig:AF_T3}
% \vspace{-6mm}
\end{figure}

\begin{figure}[t!]
\centering
\begin{subfigure}[b]{0.7\linewidth}
\centering
\includegraphics[width=\linewidth]{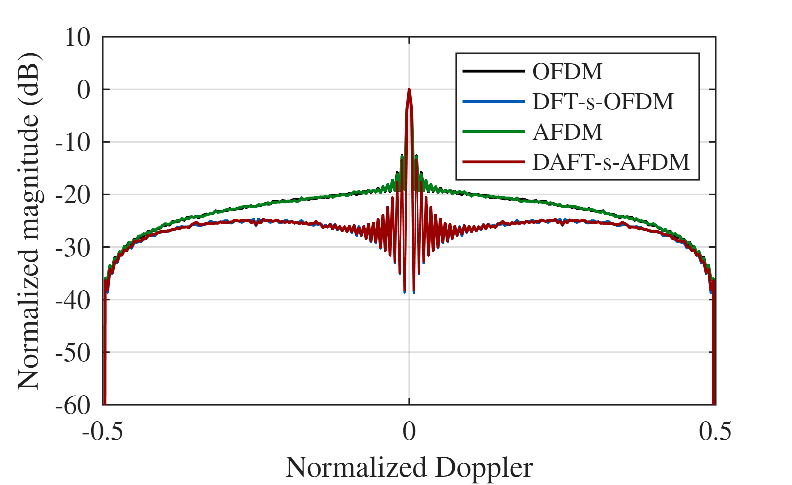}
\caption{L-FDMA ($S=1$).}
\label{fig:zero-Delay_LFDMA}
\end{subfigure}

% 第二行
\begin{subfigure}[b]{0.7\linewidth}
\centering
\includegraphics[width=\linewidth]{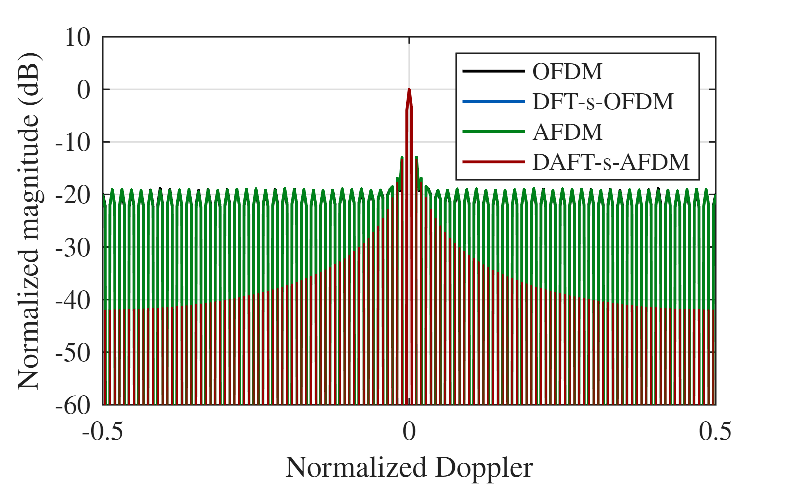}
\caption{I-FDMA ($S=N/M$).}
\label{fig:zero-delay_IFDMA}
\end{subfigure}

\caption{The zero-delay slice.}
\label{fig:AF_waveform_zero_delay}
\vspace{-6mm}
\end{figure}

\begin{figure*}[t!]
\centering
\begin{subfigure}[b]{0.32\textwidth}
\centering
\includegraphics[width=\linewidth]{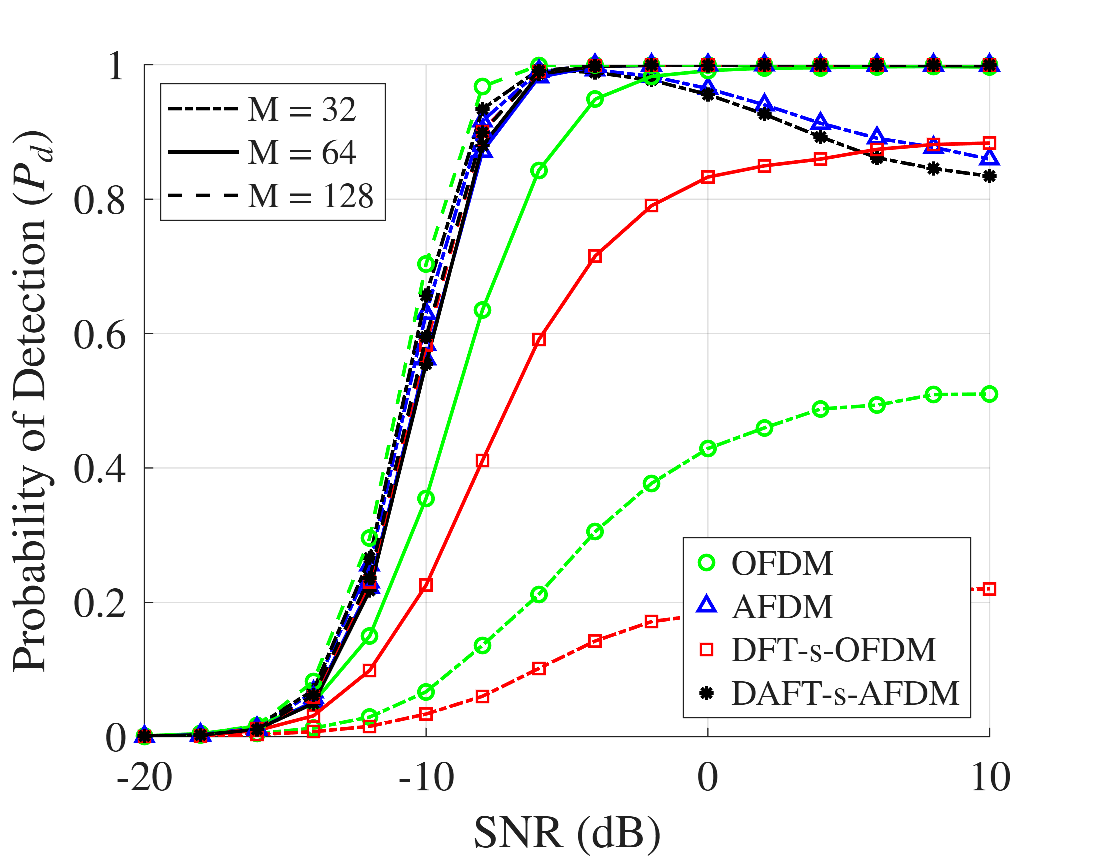}
\caption{Single target at 300 m.}
\label{fig:CFAR_correct_300_all}
\end{subfigure}
\hfill
\begin{subfigure}[b]{0.32\textwidth}
\centering
\includegraphics[width=\linewidth]{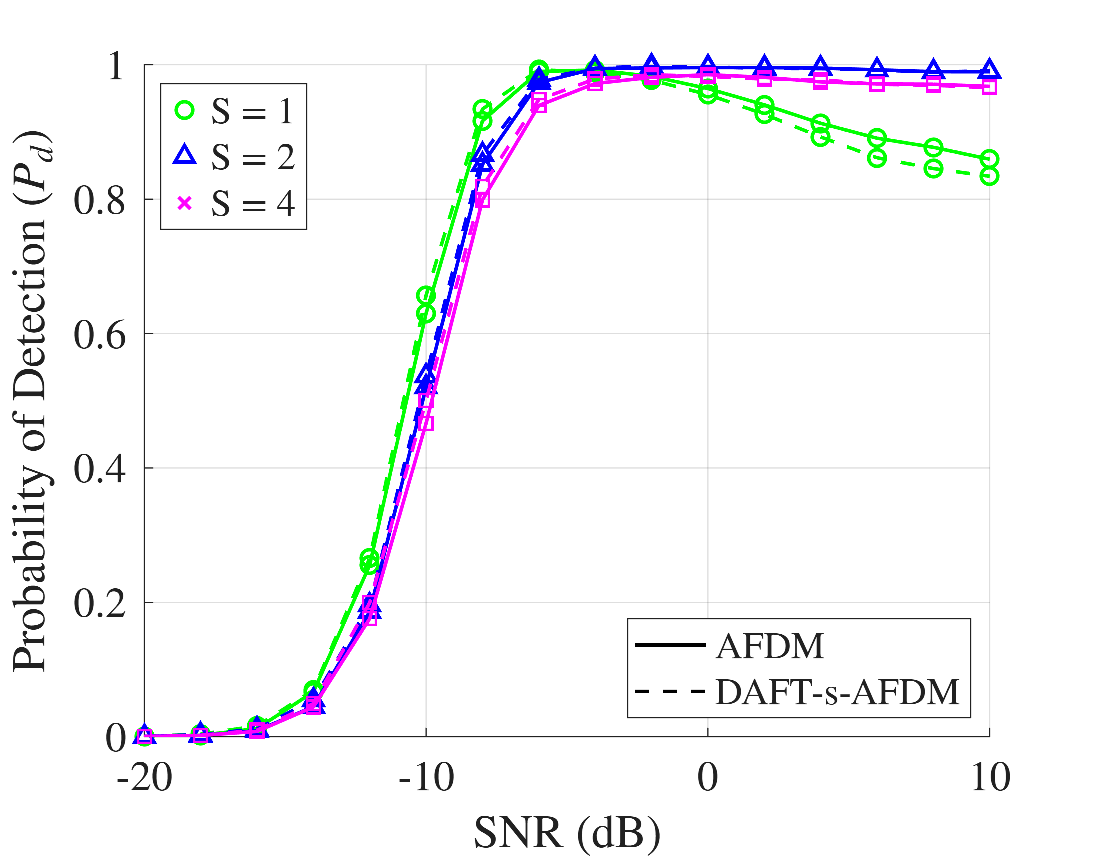}
\caption{Single target at 300 m.}
\label{fig:CFAR_correct_300_all_S}
\end{subfigure}
\hfill
\begin{subfigure}[b]{0.32\textwidth}
\centering
\includegraphics[width=\linewidth]{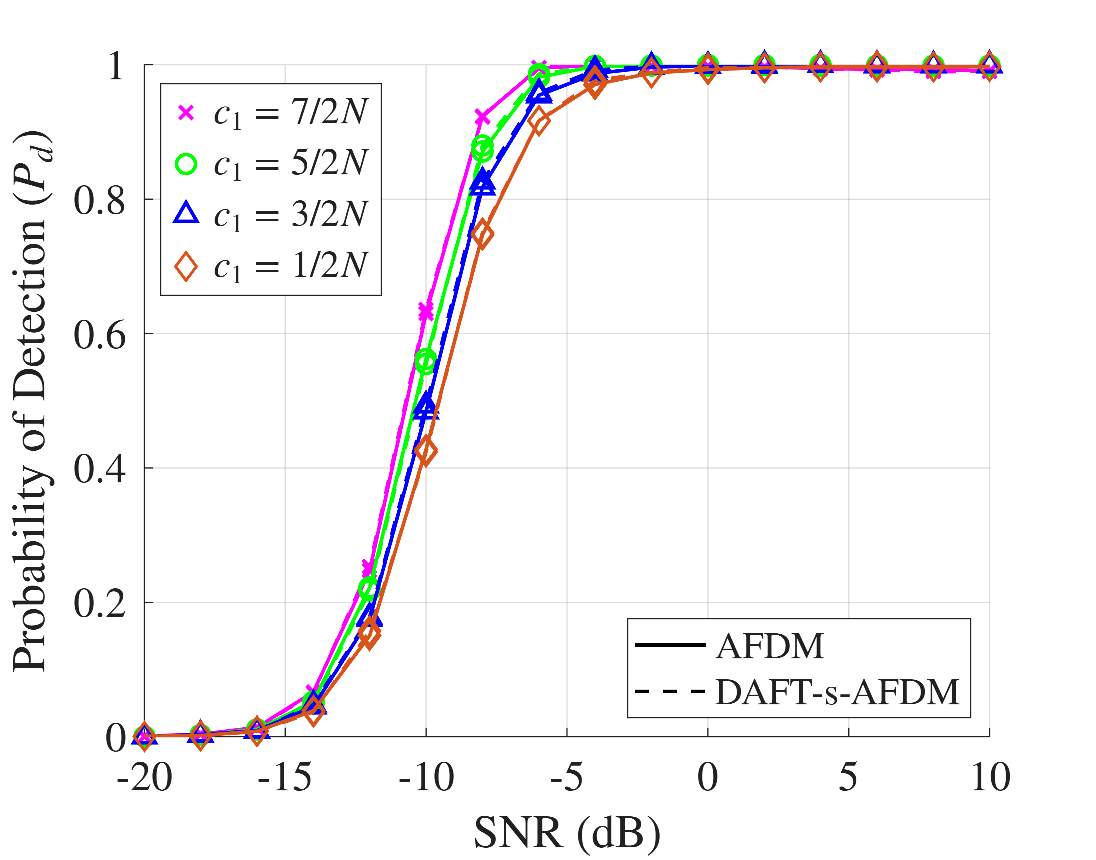}
\caption{Single target at 300 m.}
\label{fig:CFAR_correct_300_all_c1}
\end{subfigure}

\vspace{0.6em}

\begin{subfigure}[b]{0.32\textwidth}
\centering
\includegraphics[width=\linewidth]{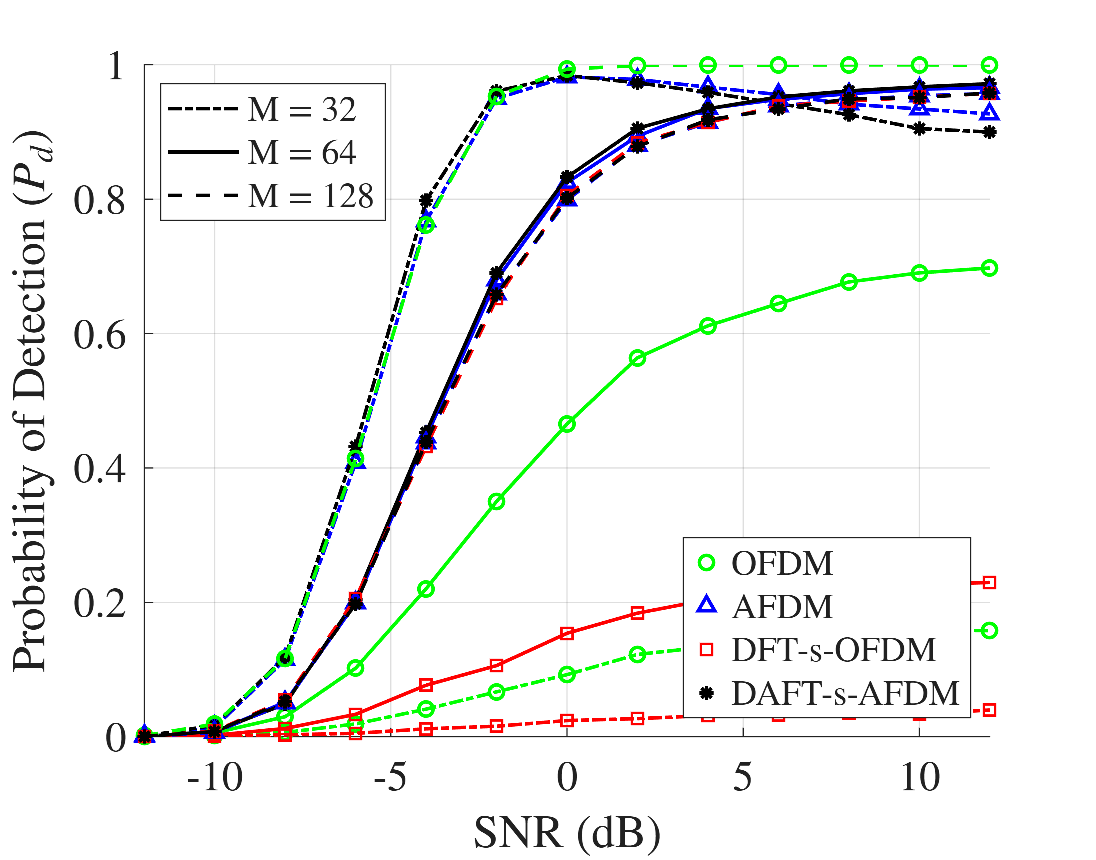}
\caption{Two targets at 300 m and 320 m.}
\label{fig:CFAR_correct_300_320_all}
\end{subfigure}
\hfill
\begin{subfigure}[b]{0.32\textwidth}
\centering
\includegraphics[width=\linewidth]{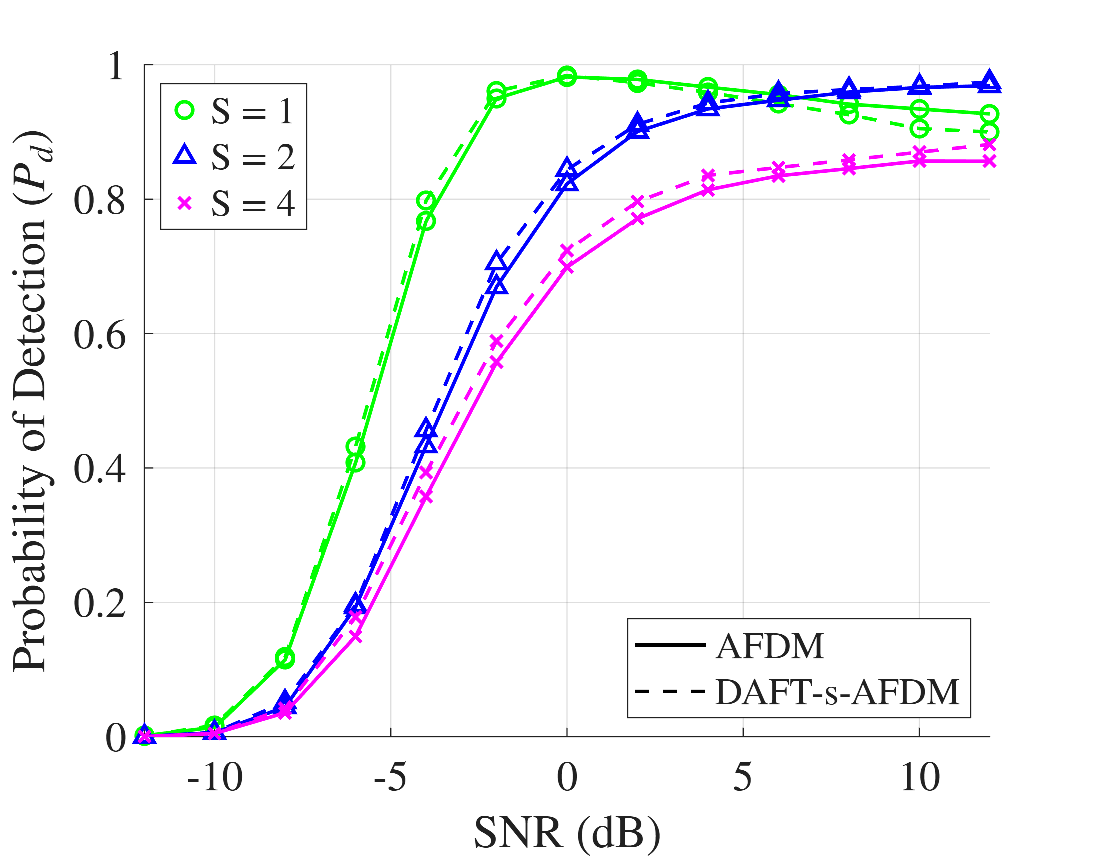}
\caption{Two targets at 300 m and 320 m.}
\label{fig:CFAR_correct_300_320_all_S}
\end{subfigure}
\hfill
\begin{subfigure}[b]{0.32\textwidth}
\centering
\includegraphics[width=\linewidth]{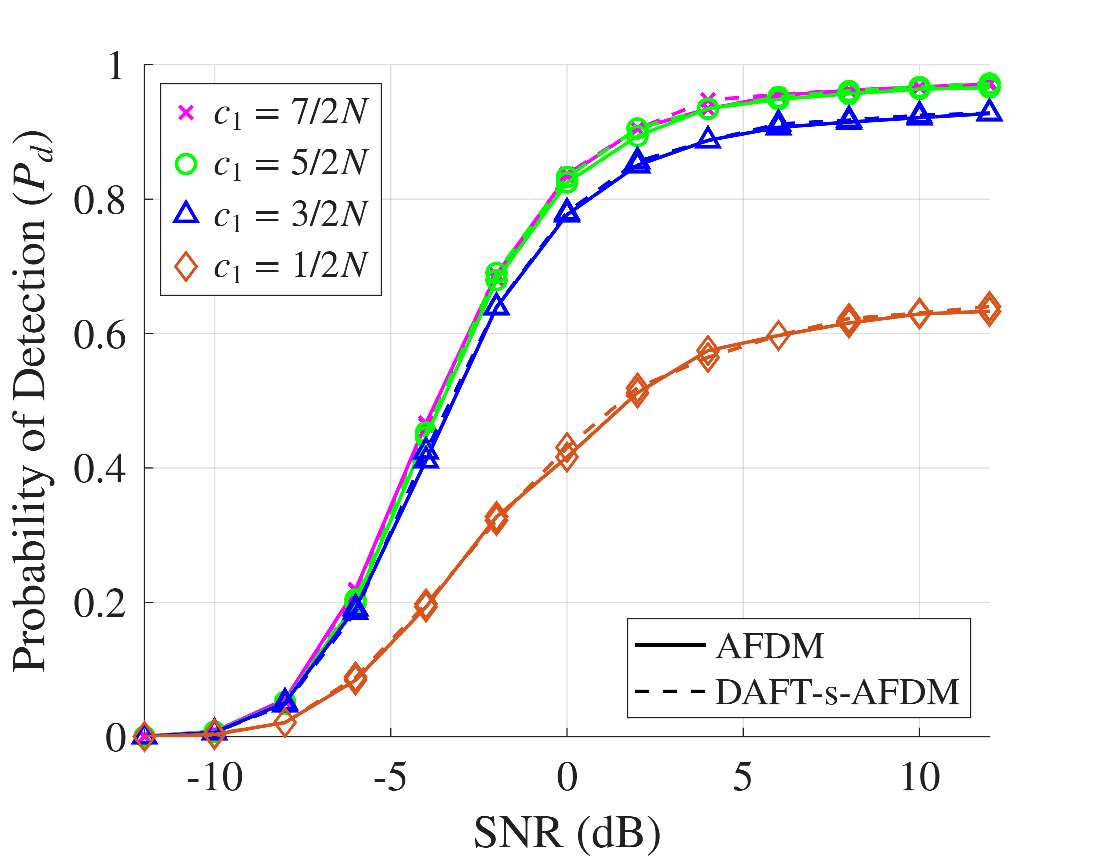}
\caption{Two targets at 300 m and 320 m.}
\label{fig:CFAR_correct_300_320_all_c1}
\end{subfigure}

\caption{CA-CFAR detection results with different $M$, $S$, and $c_1$.}
\label{fig:cfar_correct_result}
% \vspace{-5mm}
\end{figure*}

The DAAF is considered next. Theoretically, the DAAF of DAFT-s-AFDM is equivalent to the DPAF multiplied by a window function that attenuates along the delay axis. Fig.~\ref{fig:af_decomp} illustrates two slices of the DAAF for DAFT-s-AFDM under the waveform configuration $\mathrm{(b)}$, employing parameters identical to those in the simulations for Fig.~\ref{fig:AF_abc}. Figs. \ref{fig:af_decomp_range} and \ref{fig:af_decomp_velocity} reveal that the fourth-order moment of the constellation exerts a relatively small influence on the zero-Doppler cut, whereas its effect on the zero-delay cut is substantial. This finding is consistent with our analytical results in Section \ref{sec:sensing_performance}.

Fig. \ref{fig:AF_T3} illustrates the 2D AF of the $\mathcal{T}_2$ term in \eqref{eq:DAFT_AF}, where $c_1=0$ corresponds to the special case of DFT-s-OFDM. It can be observed that when $\Delta \lambda =0$, the ambiguity region influenced by the normalized fourth-order moment of the constellation is predominantly concentrated around the zero-delay slice. This reflects the inherent nature of quasi-single-carrier waveforms, such as DAFT-s-AFDM. When the conditions of Proposition \ref{prop:1} are not satisfied, the influence of the $\mathcal{T}_3$ is distributed across the entire delay-Doppler plane, resulting in multiple local peaks, as v in Fig. \ref{fig:AF_T3_delta_lambda} and \ref{fig:AF_T3_DFT_delta_lambda}. This characteristic can be exploited to improve the SNR of responses from localized targets, consistent with the behavior observed in AFDM \cite{AFDM_AF_FANLIU}. Notably, DFT-s-OFDM with $ c_1 = 0 $ exhibits the same property.

We then proceed to compute and plot the zero-delay slice for four distinct waveforms: OFDM, AFDM, DFT-s-OFDM, and DAFT-s-AFDM, all employing QPSK modulation with identical subcarrier mapping. As shown in Fig. \ref{fig:AF_waveform_zero_delay}, both DAFT-s-AFDM and DFT-s-OFDM exhibit similar characteristics in the zero-delay slice and outperform the multicarrier waveforms OFDM and AFDM because of their lower sidelobes. Furthermore, as illustrated in Figs. \ref{fig:zero-Delay_LFDMA} and \ref{fig:zero-delay_IFDMA}, with constant-modulus constellations, the ambiguity sidelobes in L-FDMA is primarily concentrated near the origin, whereas in I-FDMA it is mainly distributed at locations far from the origin. This contrasting behavior provides useful guidance for waveform selection in practical applications.

\subsection{Analysis of Parameter Impact on Detection Probability}
\label{subsec:sensing_eval}
In this subsection, we evaluate the sensing performance of the proposed schemes, focusing on the impact of key waveform parameters: the number of active subcarriers ($M$), the subcarrier mapping interval ($S$), and the chirp parameter ($c_1$). We set a carrier frequency of 77 GHz, a sampling rate of 100 MHz, and a total of $N=128$ subcarriers. The standard 64QAM constellation is employed. Target detection is executed using a matched filter combined with the cell averaging constant false alarm rate (CA-CFAR) algorithm \cite{CA_CFAR}. All presented results denote statistical averages derived from 10,000 Monte Carlo trials.

We also benchmark those four communication waveforms. To ensure a fair comparison, the OFDM and AFDM waveforms employ a subcarrier mapping scheme analogous to that of DAFT-s-AFDM, despite deviations from standard implementations. Fig. \ref{fig:cfar_correct_result} depicts the detection probability of these waveforms under various $M$ configurations, utilizing L-FDMA mapping with $c_1 = 5/2N$. Specifically, Figs. \ref{fig:CFAR_correct_300_all} and \ref{fig:CFAR_correct_300_320_all} illustrate the single-target and dual-target scenarios, respectively. In the dual-target case, the response power of the secondary target is set to half that of the primary target.

As shown in \ref{fig:CFAR_correct_300_all} and \ref{fig:CFAR_correct_300_320_all}, the chirp-based waveforms significantly outperform the non-chirp variants, while the performance disparity between AFDM and DAFT-s-AFDM is negligible. A comparison between the $M=64$ and $M=32$ configurations reveals distinct behaviors. In the low SNR regime, the $M=32$ configuration yields a higher detection probability than $M=64$. This is primarily attributed to the narrower main lobe and more concentrated energy of the AF at $M=32$. However, as SNR increases, the ambiguity features described in Proposition \ref{prop:3} amplify the impact of sidelobe peaks on target detection. This results in an increased false alarm rate, thereby degrading the probability of detection.

Fig. \ref{fig:CFAR_correct_300_all_S} and \ref{fig:CFAR_correct_300_320_all_S} illustrates the impact of the subcarrier mapping interval $S$ on detection probability, with $M=32$. While the $S=1$ configuration initially exhibits superior performance due to favorable main lobe characteristics, significant false alarms in high SNRs severely compromise detection accuracy. Conversely, the $S=2$ and $S=4$ configurations are more robustness, with $S=2$ outperforming $S=4$.

Finally, Fig. \ref{fig:CFAR_correct_300_all_c1} and \ref{fig:CFAR_correct_300_320_all_c1} examines the effect of the chirp parameter $c_1$ on sensing performance, given $M=64$ and $S=1$. A generally positive correlation is observed between the detection probability and $c_1$. Furthermore, the performance gain associated with a larger $c_1$ is more pronounced in dual-target scenarios. This phenomenon arises because $c_1$ governs the number of periodic sidelobes within the maximum delay range. A higher value of $c_1$ produces more sidelobes, effectively narrowing the widths of both the main lobe and sidelobes. This narrowing effect consequently enhances detection sensitivity at the target response.

\begin{figure}[t!]
\centering
% It's good practice to have a block diagram for the model
\includegraphics[width=0.85\linewidth]{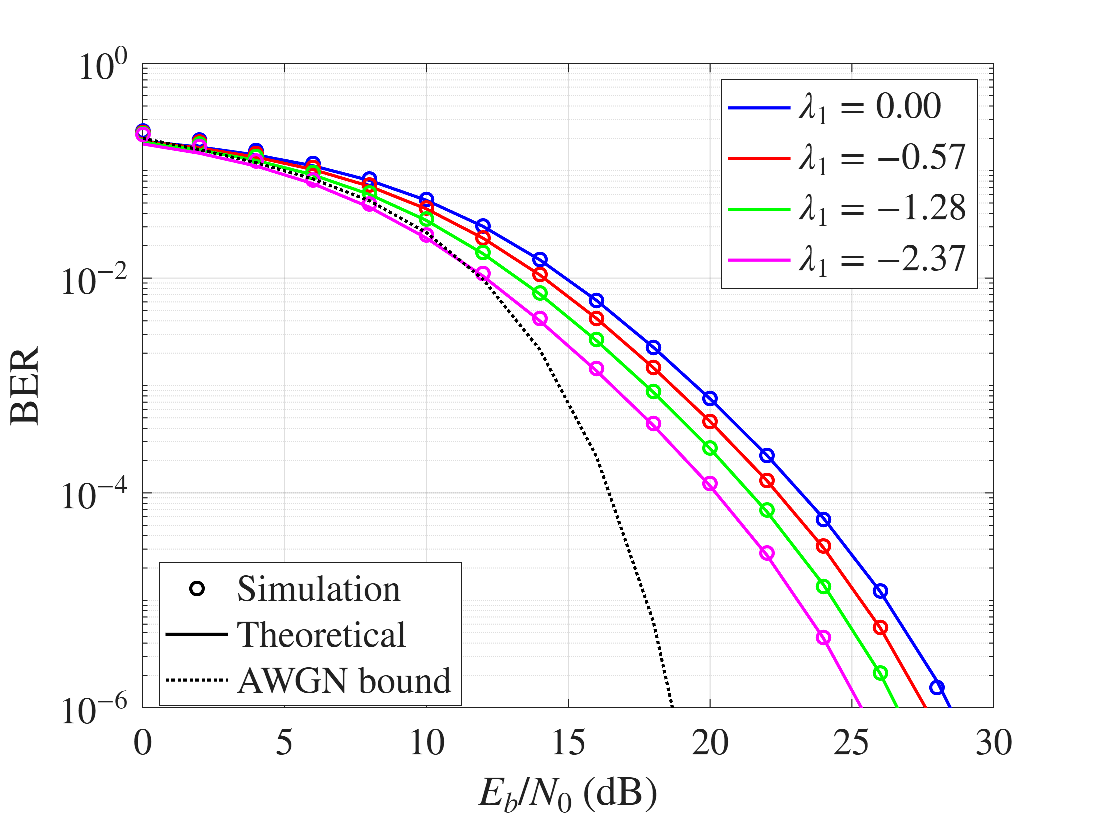}
\caption{DAFT-s-AFDM BER Comarison with PCS.}
\label{fig:BER_SIM_PCS}
% \vspace{-4mm}
\end{figure}

\begin{figure}[t!]
\centering

\includegraphics[width=0.85\linewidth]{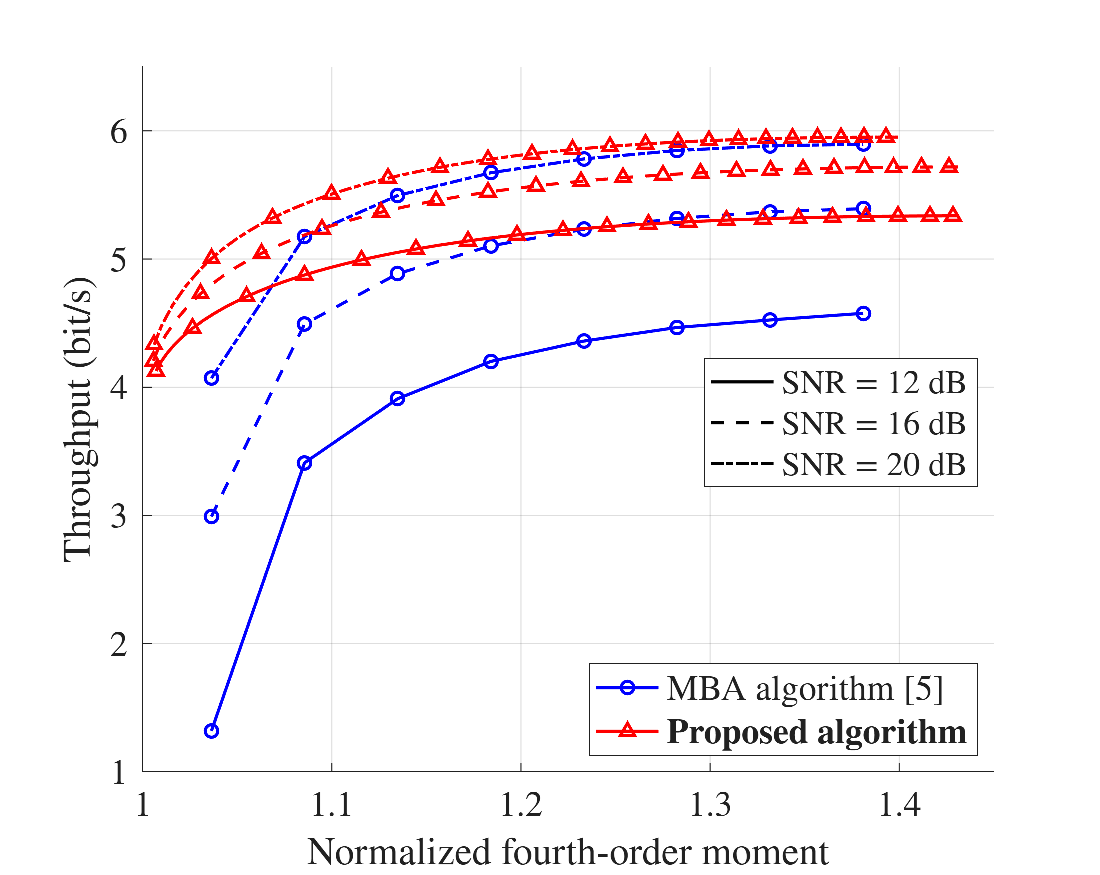}
\caption{Pareto fronts at different SNRs.}
\label{fig:Pareto}
% \vspace{-4mm}
\end{figure}

\begin{figure}[t!]
\centering
\captionsetup[subfigure]{justification=centering}

\begin{subfigure}[b]{0.32\linewidth}
\includegraphics[width=\linewidth]{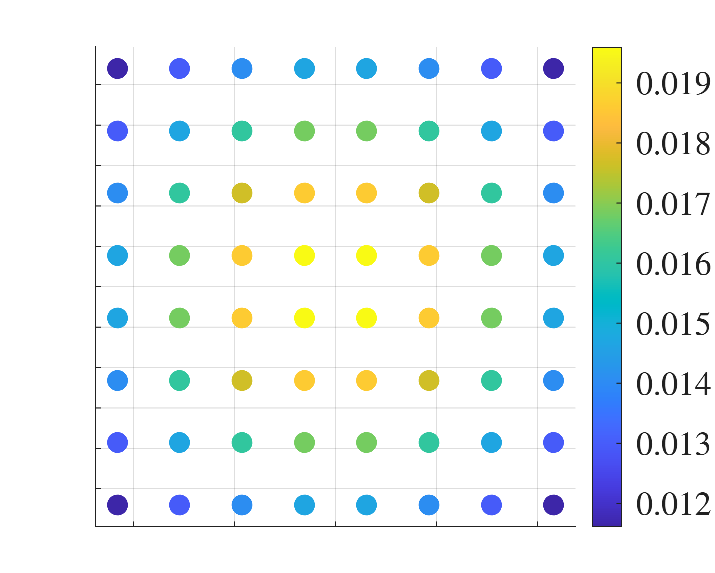}
\caption{Comm.-centric\\ $\omega_1 = 1$}
\label{fig:const_comm}
\end{subfigure}
\hfill
\begin{subfigure}[b]{0.32\linewidth}
\includegraphics[width=\linewidth]{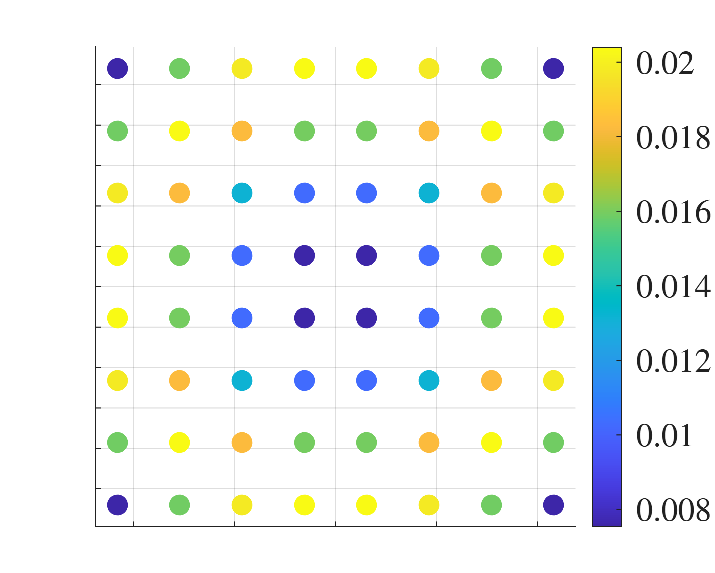}
\caption{Balanced\\ $\omega_1 = 0.5$}
\label{fig:const_balanced}
\end{subfigure}
\hfill
\begin{subfigure}[b]{0.32\linewidth}
\includegraphics[width=\linewidth]{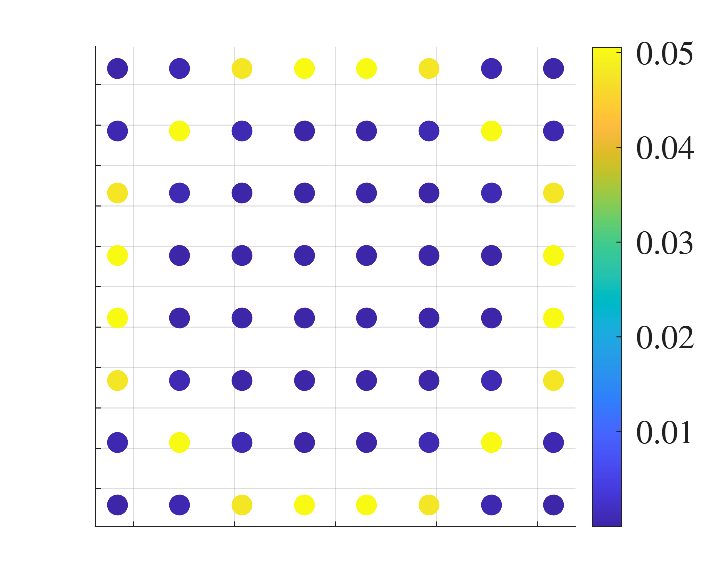}
\caption{Sens.-centric\\ $\omega_1 = 0$}
\label{fig:const_sensing}
\end{subfigure}

\caption{Optimal PCS constellations at three representative operating points on the Pareto front in Fig.~\ref{fig:Pareto}.}
\label{fig:constellations}
% \vspace{-3mm}
\end{figure}

% \subsection{Evaluation of Proposed Algorithms}
\subsection{Evaluation of Proposed Algorithm}

In this subsection, we optimize the constellation probabilities using the proposed algorithm. For the DAFT-s-AFDM, the parameters are configured as $N=128$, $M=64$, $S=1$, and $c_1=5/(2N)$. The simulation employs a three-path communication channel model with a maximum normalized Doppler shift of $\alpha=1$, and the total number of simulated bits is $10^7$. Fig.~\ref{fig:BER_SIM_PCS} presents the BER curves under the generalized MB distribution for $|\lambda_1|=0$, $0.57$, $1.28$ and $2.37$. The theoretical BER closely matches the simulation results. Moreover, as $|\lambda_1|$ increases, the overall constellation becomes more Gaussian-like, leading to a lower BER.

Building on this BER validation, the overall throughput-sensing tradeoff on the Pareto front is next examined. As observed in Fig.~\ref{fig:Pareto}, the proposed algorithm exhibits a clear performance advantage over the MBA algorithm. A small search grid with $|\mathcal{W}_{\text{grid}}| \times |\mathcal{P}_{\text{grid}}| = 4\times4$ is used. Table \ref{tab:runtime_comparison} presents the runtime comparison between the proposed algorithm and the MBA algorithm under different algorithmic parameter settings. For the MBA algorithm, the number of Monte Carlo samples is chosen as 5000, consistent with \cite{Reshaping}, and 3000 as a lower-complexity setting. It can be observed that the proposed algorithm achieves a runtime approximately one order of magnitude lower than that of the MBA algorithm. Fig.~\ref{fig:constellations} shows three constellations selected from the Pareto front at $\text{SNR}=12\text{ dB}$ by solving the PCS optimization with different trade-off weights $\omega=1$, $0.5$ and $0$. The communication-centric point lies near the maximum-throughput end, the sensing-centric point lies near the minimum normalized fourth-order moment end, and the balanced point is taken from the middle, providing a controllable tradeoff between throughput and sensing performance.

Finally, the velocity estimation performance of the three constellation schemes depicted in Fig. \ref{fig:constellations} is evaluated. In this process, the received signal matrix is derived using the cyclic cross-correlation (CCC) algorithm, followed by the application of the multiple signal classification (MUSIC) algorithm to achieve super-resolution velocity estimation \cite{Signal_process, CCC}. The simulation scenario consists of a single target with range and velocity parameters set to $[500\text{m}, 20\text{m/s}]$. Fig. \ref{fig:MUSIC_LFDMA} presents the results averaged over 1000 Monte Carlo trials. As observed in Fig.~\ref{fig:MUSIC_sub1}, the sensing-centric constellation achieves an improvement of approximately 3.1 dB in the sidelobe level compared to the communication-centric scheme. Moreover, Fig.~\ref{fig:MUSIC_sub2} plots the average sidelobe level of MUSIC-based velocity estimation versus SNR for different waveforms. It can be observed that PCS provides consistent gains across all SNR values, and DAFT-s-AFDM exhibits both the highest PCS gain and the lowest sidelobe level, since DAFT spreading lowers PAPR and AFDM is more resilient to Doppler over OFDM, allowing the same PCS shaping to deliver more significant sidelobe suppression and a lower estimation error floor. DFT-s-OFDM, as a single-carrier-like waveform, also exhibits strong velocity-estimation capability, but its non-chirp structure causes a moderate performance loss compared with DAFT-s-AFDM. Owing to the favorable time-domain autocorrelation of OFDM, its velocity-estimation performance is slightly better than that of AFDM among multicarrier waveforms, while both remain inferior to the two single-carrier-like schemes. These observations further validate our theoretical analysis and the performance superiority of the proposed algorithm.

\begin{figure}[t!]
\centering
\begin{subfigure}[t]{0.72\linewidth}
\centering
\includegraphics[width=\linewidth]{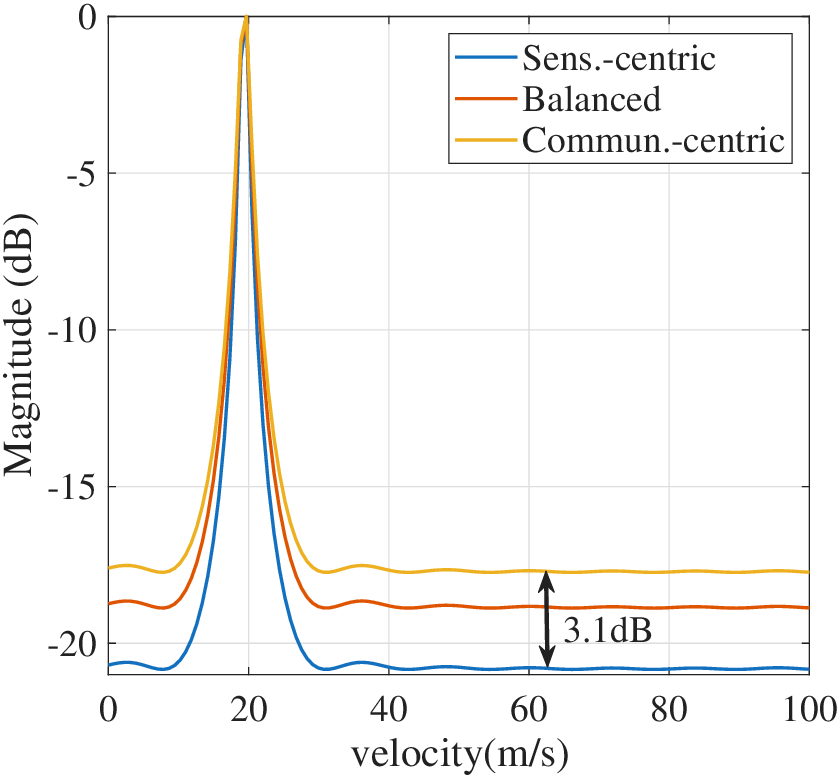}
\caption{}
% \caption{MUSIC Spectrum at $\mathrm{SNR}=20$dB under three PCS strategies.}
\label{fig:MUSIC_sub1}
\end{subfigure}
\hfill
\begin{subfigure}[t]{0.9141\linewidth}
\centering
\includegraphics[width=\linewidth]{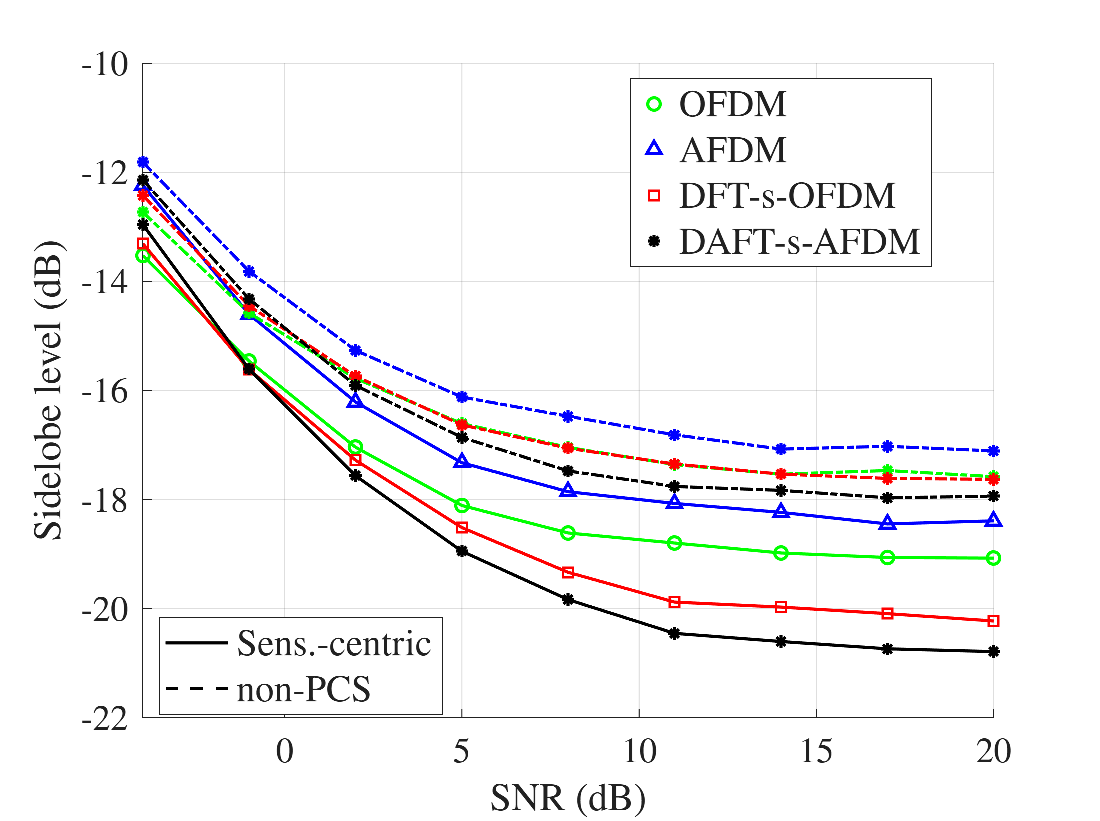}
% \caption{MUSIC spectrum sidelobe level.}
\caption{}
\label{fig:MUSIC_sub2}
\end{subfigure}
\caption{Illustration of (a) MUSIC Spectrum at SNR = 20dB under three PCS strategies, and (b) MUSIC spectrum sidelobe levels using sensing centric PCS and non-PCS strategies.}
\label{fig:MUSIC_LFDMA}
% \vspace{-6mm}
\end{figure}

\begin{table}[t]
\centering
\caption{Runtime comparison.}
\label{tab:runtime_comparison}
\begin{tabular}{lcc}
\toprule
Method & Parameter setting & Runtime (s) \\
\midrule
MBA~\cite{Reshaping}   & MC samples: 3000                & 1.66e-2 \\
MBA~\cite{Reshaping}   & MC samples: 5000                & 2.14e-2 \\
Proposed & First-stage grid: $4 \times 4$ & 2.27e-3 \\
Proposed & First-stage grid: $8 \times 8$ & 2.29e-3 \\
\bottomrule
\end{tabular}
% \vspace{-3mm}
\end{table}

% Finally, the velocity estimation performance of the three constellation schemes depicted in Fig. \ref{fig:constellations} is evaluated. In this process, the received signal matrix is derived using the cyclic cross-correlation (CCC) algorithm, followed by the application of the multiple signal classification (MUSIC) algorithm to achieve super-resolution velocity estimation \cite{Signal_process, CCC}. The simulation scenario consists of two targets with range and velocity parameters set to $[50\text{m}, 50\text{m/s}]$ and $[70\text{m}, 20\text{m/s}]$, respectively. Fig. \ref{fig:MUSIC_LFDMA} presents the results averaged over 1000 Monte Carlo trials. It is observed that the sensing-centric constellation achieves an improvement of approximately 3.69 dB in the noise floor compared to the communication-centric scheme. This result effectively validates our theoretical analysis and the performance superiority of the proposed algorithm.

% =========================================================================
% VII. CONCLUSION
% =========================================================================
\section{Conclusion}
\label{sec:conclusion}

In this paper, a comprehensive AF analysis of the DAFT-s-AFDM waveform has been presented. A simplified analytical AF expression was derived, and the corresponding properties under both L-FDMA and I-FDMA mapping schemes were characterized, showing that DAFT-s-AFDM achieves a better delay resolution than conventional DFT-s-OFDM. Based on these insights, PCS was introduced to address the communication-sensing trade-off in ISAC systems. Furthermore, a low-complexity optimization algorithm combining grid search and the simplex method was developed to design waveforms with balanced communication and sensing performance. Simulation results verified the theoretical analysis and demonstrated that the proposed PCS-based designs provide significant performance improvements, offering a flexible and promising solution for next-generation ISAC waveform design.

% =========================================================================
% INTRODUCTION (append Notations at the end)
% =========================================================================
% NOTE: The Introduction section is earlier in the document. We insert the
% Notations subsection right before the next \section after Introduction.

% =========================================================================
% REFERENCES
% =========================================================================
\bibliographystyle{IEEEtran}
\bibliography{references} % This will look for a file named "references.bib"

% =========================================================================
% DOCUMENT END
% =========================================================================
\end{document}